\documentclass[journal]{IEEEtranTCOM}
\normalsize

\usepackage{cite}

\ifCLASSINFOpdf
   \usepackage[pdftex]{graphicx}
\else
\fi

\usepackage[cmex10]{amsmath}
\usepackage{amssymb,mathalfa}
\usepackage{mathrsfs}
\usepackage{amsthm}

\usepackage{algorithm}
\usepackage{algorithmic}

\usepackage{array}

\usepackage{mdwmath}
\usepackage{mdwtab}
\usepackage{url}

\usepackage{diagbox}
\usepackage{tikz}
\usetikzlibrary{positioning}
\usetikzlibrary{arrows}
\usepackage{longtable}

\font\msbm=msbm10 at 12pt
\newcommand{\ZZ}{\mbox{\msbm Z}}

\newcommand*\mycirc[1]{%
   \begin{tikzpicture}
     \node[draw,circle,inner sep=1pt,minimum size=1.6em]{#1};
   \end{tikzpicture}}
\def \R {{\mathbb R}}
\def \Z {{\ZZ}}

\def \x {{\bf x}}
\def \y {{\bf y}}
\def \a {{\bf a}}
\def \b {{\bf b}}
\def \u {{\bf u}}
\def \v {{\bf v}}
\def \x {{\bf x}}
\def \y {{\bf y}}
\def \z {{\bf z}}

\newtheorem{theorem}{Theorem}
\newtheorem{lemma}[theorem]{Lemma}
\newtheorem{remark}[theorem]{Remark}

\newtheorem{definition}[theorem]{Definition}

\newtheorem{encoding}[theorem]{Encoding}
\newtheorem{construction}[theorem]{Construction}
\newtheorem{proposition}[theorem]{Proposition}

\begin{document}
%
\title{On Conflict Free DNA Codes}
%
%
%


\author{Krishna~Gopal~Benerjee,
Sourav Deb, 
        and~Manish~K~Gupta,~\IEEEmembership{Senior Member,~IEEE}
\thanks{Krishna Gopal Benerjee, Sourav Deb, and Manish K Gupta are with Laboratory of Natural Information Processing, Dhirubhai Ambani Institute of Information and Communication Technology Gandhinagar, Gujarat, 382007, India, e-mail: (krishna\_gopal@daiict.ac.in, sourav\_deb@daiict.ac.in and mankg@computer.org).}
}

\maketitle

\begin{abstract}
DNA storage has emerged as an important area of research. The reliability of DNA storage system depends on designing the DNA strings (called DNA codes) that are sufficiently dissimilar. In this work, we introduce DNA codes that satisfy a special constraint. Each codeword of the DNA code has a specific property that any two consecutive sub-strings of the DNA codeword will not be the same (a generalization of homo-polymers constraint). This is in addition to the usual constraints such as Hamming, reverse, reverse-complement and $GC$-content. We believe that the new constraint will help further in reducing the errors during reading and writing data into the synthetic DNA strings. We also present a construction (based on a variant of stochastic local search algorithm) to calculate the size of the DNA codes with all the above constraints, which improves the lower bounds from the existing literature, for some specific cases. Moreover, a recursive isometric map between binary vectors and DNA strings is proposed. Using the map and the well known binary codes we obtain few classes of DNA codes with all the constraints including the property that the constructed DNA codewords are free from the hairpin like secondary structures. 
\end{abstract}


%
\IEEEpeerreviewmaketitle

\ifCLASSOPTIONcompsoc
\IEEEraisesectionheading{\section{Introduction}\label{sec:introduction}}
\else
\section{Introduction}
\fi
\IEEEPARstart{T}{he} exponentially increasing demand in data storage forces to look into every possible option and DNA (DeoxyriboNucleic Acid) data storage has come out to be one of the most promising natural data storage for this purpose \cite{Jacobs2015}. 
After the first striking implementation of large-scale archival DNA-based storage architecture by Church \textit{et al.} \cite{Church1628} in 2012, followed by encoding scheme to DNA proposed by Goldman \textit{et al.} \cite{goldman2013towards} in 2013, researchers have taken great interests on the construction of DNA-based information storage systems\cite{DBLP:journals/corr/LimbachiyaG15,DBLP:journals/corr/LimbachiyaRG16} because of it's high storage density and longevity \cite{Church1628,goldman2013towards,yazdi2015rewritable}.
DNA consists of four types of bases or \textit{nucleotides} ($nt$) called adenine ($A$), cytosine ($C$), guanine ($G$) and thymine ($T$), where the Watson-Cricks complementary bases for $A$ and $C$ are $T$ and $G$ respectively and vice versa.
To store data into DNA, data need to be encoded into strings on quaternary alphabet $\{A,C,G,T\}$.  
The set of encoded DNA strings (also called DNA codewords) on the quaternary alphabet is called DNA code. 
For a DNA string, the complement is a DNA string obtained by replacing each nucleotide by it's complement.
Similarly, for a DNA string, the reverse DNA string is a DNA string in reverse order, and the complement of the reverse DNA string is called reverse-complement DNA string.
The encoded strings are synthesized using DNA synthesizer for the purpose of writing into DNA strings and the synthesized DNA strings has been stored in appropriate environment. 
To extract the source data, the stored DNA strings are read using DNA sequencing.  

During synthesis and sequencing the DNA strings, errors occur.
The errors can be reduced by choosing good encoding scheme for the DNA strings. 
Therefore, it is important to study the source of errors. 
Generally, 
insertion or deletion of repeated nucleotides occur frequently for DNA strings with consecutive repetitions of a specific nucleotide (e.g. $AC\textbf{GGGG}AT$) or of a block of nucleotides (e.g. $AG\textbf{ATATAT}GC$) up to certain length \cite{Loman,Thomson,Myers,7888471}. 
In addition, during DNA sequencing of a DNA string with consecutive identical block repetition(s), the DNA string gets misaligned more frequently \cite{Myers}. 
So, a DNA code is preferred in which each codeword does not have consecutive repetition(s) of a specific nucleotide or of a block of nucleotides. 
In this article, such DNA strings are called conflict free DNA strings. 
In literature, DNA codes without homopolymers (DNA string with consecutive repetition of a nucleotide) \cite{bornholt2016dna,blawat2016forward,2018arXiv181206798I,Erlich950,song2018codes} and without consecutive repeats of blocks \cite{7888471,8454743} are studied.
On the other hand, in this work, the considered conflict free DNA strings are not only free from homopolymers but also free from consecutive repetition of blocks of nucleotides. 

In a single stranded DNA, if there exists two sub-strands such that one is reverse-complement of another then the single stranded DNA folds back upon itself and forms antiparallel double stranded hairpin like structure (also called hairpin loop or Stem-loop) \cite{NelmsBrian,10.1007/11753681_12,1523340,2014arXiv1403.5477v1}. 
An example of such hairpin like structure is illustrated in Figure \ref{example fifure hairpin}.
For DNA sequencing, it is preferred to avoid such secondary structures \cite{10.1007/11753681_12}.
In this work, the conflict free DNA codes are constructed such that all the codewords are free from hairpin like structures with stem length more than $2$. 


		\tikzstyle{line} = [draw, -]
		\tikzstyle{A} = [draw, circle,fill=red!20, node distance=0.45cm, minimum height=0.4em]
		\tikzstyle{C} = [draw, circle,fill=blue!20, node distance=0.45cm, minimum height=0.4em]
		\tikzstyle{G} = [draw, circle,fill=green!20, node distance=0.45cm, minimum height=0.4em]
		\tikzstyle{T} = [draw, circle,fill=yellow!20, node distance=0.45cm, minimum height=0.4em]
		\tikzset{vertex/.style = {shape=circle,draw,minimum size=0em}}
		\tikzset{edge/.style = {->,> = latex'}}
		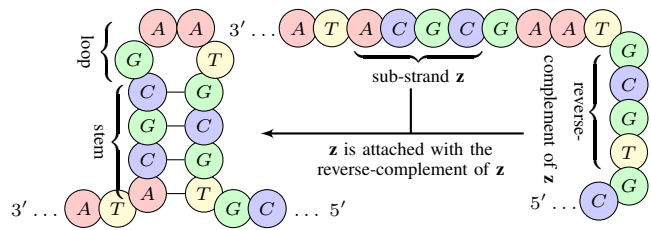
\begin{figure}
	\begin{tikzpicture}
			\scriptsize
			
			\node [A] at (-0.85,-0.2) (14) {$A$};
			\node [T] at (-0.4,-0.2) (13) {$T$};
			\node [A] (1) {$A$};
			\node [T] at (0.75,0) (2) {$T$};
		    \draw (1) -- (2);
			
			\node [C, above of=1] (3) {$C$};
			\node [G] at (0.75,0.45) (4) {$G$};
		    \draw (3) -- (4);
			
			\node [G, above of=3] (5) {$G$};
			\node [C] at (0.75,0.9) (6) {$C$};
		    \draw (5) -- (6);
			
			\node [C, above of=5] (7) {$C$};
			\node [G] at (0.75,1.35) (8) {$G$};
		    \draw (7) -- (8);
			
			\node [G] at (-0.18,1.78) (9) {$G$};
			\node [T] at (0.9,1.78) (10) {$T$};
			\node [A] at (0.14,2.2) (11) {$A$};
			\node [A] at (0.64,2.2) (12) {$A$};
			
			\node [G] at (1.14,-0.2) (15) {$G$};
			\node [C] at (1.59,-0.2) (16) {$C$};
		    
			\draw (2.3,-0.2) node {$\ldots\ 5'$};
			\draw (-1.5,-0.2) node {$3'\ldots$};
			
			\node [A] at (2,2.2) (17) {$A$};
			\node [T, right of=17] (18) {$T$};
			\node [A, right of=18] (19) {$A$};
			\node [C, right of=19] (20) {$C$};
			\node [G, right of=20] (21) {$G$};
			\node [C, right of=21] (22) {$C$};
			\node [G, right of=22] (23) {$G$};
			\node [A, right of=23] (24) {$A$};
			\node [A, right of=24] (25) {$A$};
			\node [T, right of=25] (26) {$T$};
			\node [G] at (6.4,1.9) (27) {$G$};
			\node [C, below of=27] (28) {$C$};
			\node [G, below of=28] (29) {$G$};
			\node [T, below of=29] (30) {$T$};
			\node [G, below of=30] (31) {$G$};
			\node [C] at (6,-0.1) (32) {$C$};
			\draw (1.4,2.2) node {$3'\ldots$};
			\draw (5.4,-0.1) node {$5'\ldots$};
			
			\draw (3.6,1.7) node {$\underbrace{\hspace{17mm}}_{\mbox{sub-strand \textbf{z}}}$};
			\node[label={[label distance=0.1cm,text depth=-1ex,rotate=270]{$\underbrace{\hspace{15mm}}_{\mbox{reverse-}}$}}] at (6,1) {};
			\node[label={[label distance=0.1cm,text depth=-1ex,rotate=270]{complement of \textbf{z}}}] at (5.2,0.9) {};
			\draw [edge, thick] (5,0.8) to (1.5,0.8);
			\draw [thick] (3.5,1.4) to (3.5,0.8);
			\draw (3.5,0.6) node {\textbf{z} is attached with the};
			\draw (3.5,0.3) node {reverse-complement of \textbf{z}};
			\node[label={[label distance=0.1cm,text depth=-1ex,rotate=270]{$\underbrace{\hspace{15mm}}_{\mbox{stem}}$}}] at (-0.35,0.6) {};
			\node[label={[label distance=0.1cm,text depth=-1ex,rotate=270]{$\underbrace{\hspace{8mm}}_{\mbox{loop}}$}}] at (-0.5,1.8) {};
	\end{tikzpicture}
	\caption{An example of hairpin like secondary structures in a single stranded DNA.}
	\label{example fifure hairpin}
\end{figure}

A DNA string can be read using specific hybridization between the DNA string and it's complement DNA string \cite{doi:10.1089/10665270152530818}. 
If DNA strings in a code are not different enough among themselves then nonspecific hybridization will occur and it will be a prominent cause of error. 
Therefore, a set of DNA codewords is preferred in which DNA strings are sufficiently different among themselves. 
From Metric theory, Hamming distance between two strings of same length over same alphabet is the number of positions in which the symbols in the strings are different. 
So, construction of DNA code with Hamming constraint (ensures the difference among DNA codewords), reverse constraint (ensures the difference between DNA codewords and their reverse DNA strings), and reverse-complement constraint (ensures the difference between DNA codewords and their reverse-complement DNA strings) is preferred.
In literature, DNA codes with reverse and reverse-complement constraints are constructed from finite fields and finite rings in \cite{8437313,936111,6620200,6710171}.

The thermal stability of a DNA string depends on $GC$ content (the total number of $G's$ and $C's$) in the DNA string \cite{Yakovchuk}. 
On the other hand, the high $GC$ content leads to the insertion and deletion error during polymer chain reaction (PCR). 
Therefore, such DNA codes are preferred in which each DNA codeword has the same $GC$ content and equal to almost half of it's length and the constraint for the DNA codes is called $GC$ Content constraint.
In \cite{7282568,6033808,6710171,Smith:2011:LNC:2645923.2646249}, DNA codes with balanced GC content are studied. 
DNA codes with reverse, reverse-complement and $GC$ content constraints are studied in \cite{Smith:2011:LNC:2645923.2646249,Gaborit2005LinearCF}. 
In \cite{4418465}, the lower bound on size of DNA codes with $GC$ content and reverse-complement constraints are revised. 
In fact, DNA codes with balanced GC content and without homopolymers are also studied in \cite{2018arXiv181206798I,Erlich950,song2018codes,8444440}.

A DNA code which meets multiple constraints at the same time is capable to reduce multiple type of errors efficiently during reading and writing into DNA strings. 
In literature, an algebraic solution for DNA codes with all the constraints is not studied yet. 
In this work, an algebraic structure for family of DNA codes is proposed where the constructed DNA codes meet all the constraints such as Hamming, reverse, reverse-complement, and $GC$ Content constraints.
Apart from that, all the DNA codewords do not have any consecutive identical sub-string(s) up to certain length.
In addition, these codewords are free from hairpin like secondary structures. 
In this paper, an algorithm is given which calculates the DNA code with the property that each DNA codeword does not have any consecutive repeated sub-string of any length. 
In addition, DNA codes with Hamming constraint, reverse constraint, reverse-complement constraint, and $GC$ Content constraint are obtained. 
For a DNA code with all the constraints, the obtained code size is improved for some specific parameters as given in \cite[Table I]{8424143}. 
Further, family of DNA codes have been obtained with Hamming, reverse, reverse-complement, and $GC$ Content constraints, where each DNA codeword is free from hairpin like secondary structure and repetition(s) of any consecutive identical sub-string(s) up to certain length.

In Section \ref{sec:preliminary}, preliminary for DNA codes are discussed.
Complete conflict free DNA codes with all the constraints are studied in Section \ref{sec:ccf DNA}. 
A recursive mapping from binary strings to DNA strings is discussed in Section \ref{sec:mapping}, which also is an isometry between a newly defined distance over binary strings and Hamming distance over DNA strings.
The conditions on binary strings are obtained, which ensure the constraints on encoded DNA strings in the same section. 
In Section \ref{sec:family}, a family of DNA codes are obtained from binary Reed-Muller codes.
\ref{sec:conclusion} concludes the work.

\section{preliminary} 
\label{sec:preliminary}

A code $\mathscr{C}$ 
$(n,M,d)$ over an alphabet $\Sigma$ of size $q$ is a set of $M$ distinct strings (also called codewords) each of length $n$ and the distance between any two distinct strings is atleast $d$. 
Codes over $\{0,1\}$ and $\Sigma_{DNA}$ = $\{A,C,G,T\}$ are called binary codes and DNA codes (denoted by $\mathscr{C}_{DNA}$) respectively. 
For various applications, codes with various distances (such as Gau distance \cite{8437313}) are studied in literature. 
In this work, DNA codes with Hamming distance and binary codes with a newly defined distance are studied. 
For any strings $\x$ and $\y$ in $\Sigma^n$, the Hamming distance $d_H(\x,\y)$ between the $\x$ and $\y$ is the total number of positions at which they differ. 
For a code $\mathscr{C}\subset\Sigma^n$, the minimum Hamming distance is $d_H$ = $\min\{d_H(\x,\y):\x\neq\y\mbox{ and }\x,\y\in\mathscr{C}\}$.
For a field or a ring defined on the alphabet $\Sigma$, if each codeword of a code is a linear combination of rows of some matrix over $\Sigma$ then the code is called a linear code and the matrix is called the generator matrix of the code.

For a DNA string $\x$ = $(x_1\ x_2\ldots x_n)\in\Sigma_{DNA}^n$, the reverse, complement and reverse-complement DNA strings of $\x$ are $\x^r$ = $(x_n\ x_{n-1}\ldots x_1)$, $\x^c$ = $(x_1^c\ x_2^c\ldots x_n^c)$, and $\x^{rc}$ = $(x_n^c\ x_{n-1}^c\ldots x_1^c)$ respectively. 
For any DNA code $\mathscr{C}_{DNA}$, if $\x,\y$ are DNA codewords then, the constraints on $\mathscr{C}_{DNA}$ are defined as follows \cite{doi:10.1089/10665270152530818}.
(i) Hamming constraint: the Hamming distance $d_H(\x,\y)\geq d_H$, where $\x\neq\y$.
(ii) reverse constraint: the Hamming distance $d_H(\x,\y^r)\geq d_H$, where $\x\neq\y^r$ but $\x$ may be equal to $\y$ and, $\y^r$ may not be a codeword in $\mathscr{C}_{DNA}$.
(iii) reverse-complement constraint: the Hamming distance $d_H(\x,\y^{rc})\geq d_H$, where $\x\neq\y^{rc}$ but $\x$ may be equal to $\y$, and $\y^{rc}$ may not be a codeword in $\mathscr{C}_{DNA}$.
(iv) $GC$ Content constraint: If the total number of $G$'s and $C$'s in each codeword is same and equal to $g$ then the code satisfies $g$-$GC$ content constraint. 
For a specific case $g = \lfloor n/2\rfloor$, the $\lfloor n/2\rfloor$-$GC$ content constraint is called simply $GC$ content constraint. 
Consider a DNA code $\mathscr{C}_{DNA}$ with minimum Hamming distance $d_H$.
For each $\x\in\mathscr{C}_{DNA}$, if ($i$) $\x^r\in\mathscr{C}_{DNA}$ then from the distance property of code, $d_H(\y,\x^r)\geq d_H$ for each $\y\in\mathscr{C}_{DNA}$, therefore, the code satisfies the reverse constraint. 
Similarly for each $\x\in\mathscr{C}_{DNA}$, if ($ii$) $\x^{rc}\in\mathscr{C}_{DNA}$ then again from the distance property of code, $d_H(\y,\x^{rc})\geq d_H$ for each $\y\in\mathscr{C}_{DNA}$, and hence the code satisfies the reverse-complement constraint. 
In consequence researchers are curious in construction of DNA codes which are closed under reverse and reverse-complement DNA strings \cite{8437313}.
Thus motivated, we construct set of DNA strings for given length such that those DNA strings satisfy multiple constraints. 
In the following lemma, the distinct DNA strings with multiple constraints are enumerated.
\begin{lemma}
For a given length $n$,
\begin{enumerate}
    \item there exists $4^{\lceil n/2\rceil}$ number of distinct DNA strings $\x\in\Sigma_{DNA}^{n}$ such that $\x$ = $\x^r$,  
    \item for even $n$, there exists $4^{n/2}$ number of DNA strings $\x\in\Sigma_{DNA}^n$ such that $\x$ = $\x^{rc}$,
    \item for a positive integer $m$ ($\leq n$), there exists $\binom{n}{m}2^n$ distinct DNA strings $\x\in\Sigma_{DNA}^{n}$ each with $GC$ content $m$, 
    \item for even positive integer $m$ ($\leq n$), there exists $\binom{n/2}{m/2}2^{n/2}$ distinct DNA strings $\x\in\Sigma_{DNA}^{n}$ each with $GC$ content $m$ and $\x$ = $\x^{rc}$, where $n$ is even, and
    \item for positive integer $m$ ($\leq n$), there exists $\eta$ distinct DNA strings $\x\in\Sigma_{DNA}^{n}$ each with $m$ - $GC$ content and $\x$ = $\x^r$, where
    \[
    \eta = \left\{
    \begin{array}{ll}
        0 & \mbox{if }n\mbox{ is even and }m\mbox{ is odd}, \\
        \binom{\lfloor n/2\rfloor}{\lfloor m/2\rfloor}2^{\lceil n/2\rceil} & \mbox{otherwise.}
    \end{array}\right.
    \] 
\end{enumerate}
\end{lemma}
\begin{proof}
Consider $\x$ = $(x_1\ x_2\ldots x_n)\in\Sigma_{DNA}^n$.
\begin{itemize}
    \item[1)] \ If $\x$ = $\x^r$ then, for $i=1,2,\ldots,\lceil n/2\rceil$, $x_i$ = $x_{n-i+1}$. 
Therefore, there exists $4^{\lceil n/2\rceil}$ number of distinct DNA strings $\x\in\Sigma_{DNA}^{n}$ such that $\x$ = $\x^r$.
    \item[2)] \ If $\x$ = $\x^{rc}$ then, for $i=1,2,\ldots,\lfloor n/2\rfloor$, $x_i$ = $x_{n-i+1}^c$. 
Therefore, there exists $4^{n/2}$ number of distinct DNA strings $\x\in\Sigma_{DNA}^{n}$ such that $\x$ = $\x^{rc}$.
Note that, for any positive odd $n$, any $\x\in\Sigma_{DNA}^n$ can not equal $\x^{rc}$.
    \item[3)] \ There are $\binom{n}{m}2^m$ ways to fill $m$ positions out of $n$ positions by a symbol from $\{C,G\}$.
If the remaining $n$ - $m$ positions are filled by a symbol from $\{A,T\}$ then there are $\binom{n}{m}2^n$ distinct DNA strings each with $GC$ content $m$.
\end{itemize}
The proof of remaining results are similar. 
\end{proof}
For theoretical analysis, we define complement constraint for a DNA code similar to reverse and reverse-complement constraints.
A DNA code $\mathscr{C}_{DNA}$ satisfies the complement constraint, if, for any $\x$ and $\y$ in $\mathscr{C}_{DNA}$, $d_H(\x,\y^c)\geq d_H$.
From the triangular property and definition of Hamming distance, one can observe the following Remark. 
\begin{remark}
A DNA code with reverse and complement constraints will also satisfy reverse-complement constraint.
\label{r rc constraint remark}
\end{remark}
Note that for any codeword $\x$ in $\mathscr{C}_{DNA}$, if $\x^r$, $\x^c\in\mathscr{C}_{DNA}$ then $\x^{rc}\in\mathscr{C}_{DNA}$.
Therefore, from the definition of a code with minimum Hamming distance $d_H$, for any codeword $\x$ in $\mathscr{C}_{DNA}$, if $\x^r$, $\x^c\in\mathscr{C}_{DNA}$ then the DNA code $\mathscr{C}_{DNA}$ is closed under Hamming, reverse and reverse-complement constraints with the minimum Hamming distance $d_H$. 

In a single stranded DNA, if there exist two sub-strands such that one sub-strand is the reverse-complement of another sub-strand then the DNA strand fold back and attach the both sub-strands to each other and forms hairpin like secondary structures with stems and loops of certain length. 
Short loops which are less than three bases long are highly unstable. 
The stem size of more than $2$ bases long reasonably approximates the hairpin like structures. 
The single stranded DNA $AT\textbf{ACGC}GAAT\textbf{GCGT}GC$, considered in Figure \ref{example fifure hairpin}, contains the reverse-complementary sub-strands $ACGC$ and $TGCG$ (see the bold sub-strands). 
The sub-strands are attached to each other and forms a stem of length $4$ base pairs, and a loop of size $4$ bases long. 
One can easily observe that, for given positive integers $n$ and $m$ ($\leq n/2$), consider DNA strings of length $n$ reduces as $m$ increases. such that each DNA string does not contain two sub-strings each of size more than $m-1$ and both the sub-strings are reverse-complement of each other. 
The total number of such DNA strings reduces as $m$ decreases. 
So, it is reasonable to consider $m=3$. 
Therefore, in this work, DNA strings are considered to be free from hairpin like structures with stem length of more than $2$ bases long.
\begin{definition}
A DNA string is called free from reverse-complement sub-strings (of length more than $2$) if the DNA string does not contain any two sub-strings of length more than $2$ such that one is the reverse-complement of the other.
Clearly the DNA string does not form a hairpin like secondary structure with stem length of more than $2$.
\label{kg hairpin definition}
\end{definition}
\begin{remark}
A DNA string is free from reverse-complement sub-strings if and only if the complement of the DNA string is also free from reverse-complement sub-strings.
The same result holds for reverse and reverse-complement. 
\label{kg hairpin remark}
\end{remark}
For example, the DNA string $ACATCG$ is free from reverse-complement sub-strings because the reverse-complement of $ACA$ is $TGT$ and $TGT$ is not a sub-string of $ACATCG$. 
The reverse and reverse-complement DNA strings $GCTACA$ and $CGATGT$ are also free from reverse-complement sub-strings. 

\section{On complete conflict free DNA strings}
\label{sec:ccf DNA}
In literature, errors are frequent for the existence of homopolymers and consecutive repetition of same sub-string of certain length in DNA during synthesize and sequencing the DNA \cite{Loman,Thomson,Myers,7888471}. 
Therefore, specific construction of DNA codes are preferred which excludes any DNA codeword that contains homopolymer or consecutive consecutive repetition of same sub-string of certain length.
For such DNA codes, we define the following.
\begin{definition}
For positive integers $n$, $\ell$ ($\leq n/2$) and $t$ $\leq \ell$, a DNA string is called $\ell$ conflict free, if the DNA string is free from consecutive repetition(s) of identical sub-string(s) of length $t$ for each $t=1,2,\ldots,\ell$. 
\label{conflict free def}
\end{definition}
For example the DNA string $ATCATCG$ is $2$ conflict free because any two consecutive sub-strings of same length ($\leq 2$) are not same, $i.e.$, the DNA string does not contain any of the DNA strings $ATAT$, $TCTC$, $CACA$, $CGCG$, $AA$, $TT$, $CC$ and $GG$. 
But, $ATC$ is repeating twice in the DNA string, so the DNA string is not $3$ conflict free.
Note that, $1$ conflict free DNA strings are also known as DNA strings free from homopolymers in literature \cite{bornholt2016dna,blawat2016forward,2018arXiv181206798I,Erlich950,song2018codes}.
Also note, for positive integer $\ell$ ($\geq 2$), an $\ell$ conflict free DNA string is also $\ell-1$ conflict free.
\begin{remark}
A DNA string is $\ell$ conflict free if and only if the complement of the DNA string is also $\ell$ conflict free.
Note that the result also holds for reverse and reverse-complement DNA strings.
\label{conflict free comp}
\end{remark}
\begin{table*}[t]
\caption{List of improved lower bound of $A_4^{cf,GC,r,rc}(n,d_H,\lfloor n/2\rfloor)$ (maximum size of complete complete conflict free DNA code with Hamming, reverse, reverse-complement and $GC$ Content constraints) from \cite[Table I]{8424143} and \cite[table II]{4418465}.}
\begin{center}
\begin{tabular}{|c|c|c|c|}
\hline
DNA code parameters & Lower bound of $A_4^{cf,GC,r,rc}(n,d_H,\lfloor n/2\rfloor)$ & Lower bound of $A_4^{cf,GC}(n,d_H,1)$ & Lower bound of $A_4^{GC,rc}(n,d_H)$ \\
 $(n,d_H)$ &  Table \ref{lb of CCF code size table} & Table I in \cite{8424143} & Table II in \cite{4418465} \\ \hline
 $(4,3)$   & $12$                                   & $11$                      & $11$   \\ \hline
 $(6,4)$   & $20$                                   & $16$                      & $16$   \\ \hline
 $(8,6)$   & $12$                                   & $9$                       & $12$   \\ \hline
 $(9,6)$   & $16$                                   & $15$                      & $20$   \\ \hline
 $(9,9)$   & $2$                                    & $0$                       & $1$   \\ \hline
 $(10,7)$  & $16$                                   & $7$                       & $16$   \\ \hline
 $(10,8)$  & $8$                                    & $5$                       & $8$   \\ \hline
\end{tabular}
\end{center}
\label{optimal list}
\end{table*}  
Now, for a positive integer $\ell$, the $\ell$ conflict free DNA code is defined as following.
\begin{definition}
For positive integers $n$ and $\ell$ $(\leq\lfloor n/2\rfloor)$, a DNA code with length $n$ is called $\ell$ conflict free if each codeword of the DNA code is $\ell$ conflict free.
\end{definition}
For example, the DNA code $\left\{ACTG, TGAC, CAGT, \right.$ $\left. GTCA\right\}$ is $2$ conflict free DNA code, because each DNA codeword is $2$ conflict free DNA string. 

Many computational approach to construct DNA codes with some additional constraints are studied in literature \cite{4418465,6710171,4019650,10.1007/3-540-36440-4_20,10.1007/11816102_32}. 
In this work, $\ell$ conflict free DNA codes are constructed using stochastic local search in a seed set of $\ell$ conflict free DNA strings such that each DNA string has a fix $GC$ content constraint.
The computational construction for $\ell$ conflict free DNA codes is given as follows. 
\begin{construction}
For given positive integers $n$, $\ell$ ($\leq \lfloor n/2\rfloor$) and $g$ ($\leq n$), let $\mathcal{S}\in\Sigma_{DNA}^n$ be the set of all $\ell$ conflict free DNA strings such that $GC$ content of each DNA string is $g$. 
For a sub-set $R$ of random cardinality and containing DNA strings which are randomly selected from $\mathcal{S}$, compute $\mathscr{C}_{DNA}$ = $R\cup\{\x^r,\x^c:\x\in\mathscr{C}\}$.
\label{lb of CCF code size cons}
\end{construction}

For example, consider $n$ = $3$, $\ell$ = $1$ and $g$ = $2$.
The seed set will be $\mathcal{S}$ = $\left\{ACG, AGC, CAC, CAG, CGA, CGT, CTC, \right.$ $\left. CTG, GAC, GAG, GCA, GCT, GTC, GTG, TCG, TGC \right\}$.
Note that, for $R$ = $\{CAC, CGT, ACG, TGC\}$, the $1$ conflict free DNA code with $GC$ content, reverse and reverse-complement constraints is $\mathscr{C}_{DNA}$ = $\left\{CAC, CGT, ACG, \right.$ $\left. TGC, GCA, GTG\right\}$ of code size $M$ = $6$ and minimum Hamming distance $d_H$ = $2$. 

\begin{table*}[t]
\centering
\caption{Lower bound of $A_4^{cf,GC,r,rc}(n,d_H,\ell)$ (maximum size of $\ell$ complete conflict free DNA code with Hamming, reverse, reverse-complement and $GC$ Content constraints). The values which are meeting existing values or improved values from \cite[Table I]{8424143} or \cite[table II]{4418465} are written in bold font.}
\begin{tabular}{|c|c||c|c|c|c|c|c|c|c|c|c|}
\hline
 & &\multicolumn{10}{c|}{$d_H$} \\ \cline{3-12}
 $n$    & $\ell$  & $1$           & $2$            & $3$           & $4$          & $5$           & $6$           & $7$          & $8$          & $9$  & $10$        \\ \hline\hline
 $2$    & $1$ & $8$     & $4$           & -              & -             & -            & -             & -             & -            & -            & -            \\ \hline
 $3$    & $1$ & $16$    & $6$           & $2$            & -             & -            & -             & -             & -            & -            & -            \\ \hline
 $4$    & $2$ & $48$    & $\textbf{32}$ & $\mycirc{\textbf{12}}$  & $\textbf{4}$  & -            & -             & -             & -            & -            & -            \\ \cline{2-12}
        & $1$ & $56$    & $\textbf{32}$ & $\mycirc{\textbf{12}}$ & $\textbf{4}$ & - & - & - & - & - & - \\ \hline
 $5$    & $2$ & $108$   & $48$          & $14$           & $4$           & $\textbf{2}$ & -             & -             & -            & -            & -            \\ \cline{2-12}
        & $1$ & $128$   & $52$          & $14$           & $4$           & $\textbf{2}$ & - & - & - & - & - \\ \hline
 $6$    & $3$ & $320$   & $88$          & $32$           & $\mycirc{\textbf{20}}$ & $0$          & $\textbf{4}$  & -             & -            & -            & -            \\ \cline{2-12}
        & $2$ & $320$   & $88$          & $32$           & $\mycirc{\textbf{20}}$ & $0$          & $\textbf{4}$  & - & - & - & - \\ \cline{2-12}
        & $1$ & $424$   & $100$         & $32$           & $\mycirc{\textbf{20}}$ & $0$          & $\textbf{4}$  & - & - & - & - \\  \hline
 $7$    & $3$ & $704$   & $136$         & $48$           & $28$          & $10$         & $\textbf{4}$  & $\textbf{2}$  & -            & -            & -            \\ \cline{2-12}
        & $2$ & $740$   & $142$         & $50$           & $28$          & $10$         & $\textbf{4}$           & $\textbf{2}$           & - & - & - \\ \cline{2-12}
        & $1$ & $1040$  & $146$         & $50$           & $28$          & $10$         & $\textbf{4}$ & $\textbf{2}$ & - & - & - \\  \hline
 $8$    & $4$ & $2024$  & $236$         & $76$           & $48$          & $20$         & $\mycirc{\textbf{12}}$ & $0$           & $\textbf{4}$ & -            & -            \\ \cline{2-12}
        & $3$ & $2064$  & $238$         & $86$           & $48$          & $20$ & $\mycirc{\textbf{12}}$ & $0$ & $\textbf{4}$ & - & - \\ \cline{2-12}
        & $2$ & $2192$  & $248$         & $88$           & $48$ & $20$ & $\mycirc{\textbf{12}}$ & $0$ & $\textbf{4}$ & - & -\\ \cline{2-12}
        & $1$ & $3352$  & $252$         & $92$           & $48$ & $20$ & $\mycirc{\textbf{12}}$ & $0$ & $\textbf{4}$ & - & - \\ \hline
 $9$    & $4$ & $4568$  & $336$         & $112$          & $60$          & $24$         & $\mycirc{\textbf{16}}$ & $6$           & $\textbf{4}$ & $\mycirc{\textbf{2}}$ & -            \\ \cline{2-12}
        & $3$ & $4660$  & $344$ & $116$ & $60$  & $24$ & $\mycirc{\textbf{16}}$ & $6$ & $\textbf{4}$ & $\mycirc{\textbf{2}}$ & - \\ \cline{2-12}
        & $2$ & $5144$  & $346$ & $116$ & $60$  & $24$ & $\mycirc{\textbf{16}}$ & $6$ & $\textbf{4}$ & $\mycirc{\textbf{2}}$ & - \\ \cline{2-12}
        & $1$ & $8576$  & $386$ & $120$ & $64$  & $28$ & $\mycirc{\textbf{16}}$ & $6$ & $\textbf{4}$ & $\mycirc{\textbf{2}}$ & - \\  \hline
 $10$   & $5$ & $13008$ & $564$ & $184$ & $92$  & $48$ & $28$ & $\mycirc{\textbf{16}}$ & $\mycirc{\textbf{8}}$ & $0$ & $\textbf{4}$ \\ \cline{2-12}
        & $4$ & $13008$ & $564$ & $184$ & $92$  & $48$ & $28$ & $\mycirc{\textbf{16}}$ & $\mycirc{\textbf{8}}$ & $0$ & $\textbf{4}$ \\ \cline{2-12}
        & $3$ & $13424$ & $568$ & $192$ & $92$  & $48$ & $28$ & $\mycirc{\textbf{16}}$ & $\mycirc{\textbf{8}}$ & $0$ & $\textbf{4}$ \\ \cline{2-12}
        & $2$ & $15104$ & $596$ & $196$ & $100$ & $48$ & $28$ & $\mycirc{\textbf{16}}$ & $\mycirc{\textbf{8}}$ & $0$ & $\textbf{4}$ \\ \cline{2-12}
        & $1$ & $27208$ & $660$ & $208$ & $104$ & $48$ & $28$ & $\mycirc{\textbf{16}}$ & $\mycirc{\textbf{8}}$ & $0$ & $\textbf{4}$ \\  \hline
\end{tabular}
\label{lb of CCF code size table}
\end{table*}  
For a given DNA string, the computational complexity to determine whether the DNA string is $\ell$ conflict free (using Definition \ref{conflict free def}) is more than the computational complexity to determine the $GC$ content of the DNA string.
Therefore, in order to construct the seed set for the Construction \ref{lb of CCF code size cons}, one can reduce the computations by removing DNA strings without $GC$ content $g$ first and than DNA strings which are not $\ell$ conflict free from the complete set $\Sigma_{DNA}^n$.

For given length $n$ and distance $d$, the maximum size of code is subject to interest among researches.
Now, similar to \cite{doi:10.1089/10665270152530818}, some notations for maximum size of $\ell$ conflict free DNA codes are introduced here. 
Let $A_4^{cf}(n,d_H,\ell)$ is the maximum size of a DNA code with length $n$ and minimum Hamming distance $d_H$ such that each DNA codeword is $\ell$ conflict free.
Similarly, for $\ell$ conflict free DNA codes, $A_4^{cf,r}(n,d_H,\ell)$, $A_4^{cf,rc}(n,d_H,\ell)$ and $A_4^{cf,GC}(n,d_H,\ell)$ denote the maximum size of $\ell$ conflict free DNA code with reverse,  reverse-complement and $GC$ content constraints respectively. 
The maximum size of the $\ell$ conflict free DNA code with $GC$ content, reverse and reverse-complement constraints is denoted as $A_4^{cf,GC,r,rc}(n,d_H,\ell)$.

The relation among the size of $\ell$ conflict free DNA codes with additional constraints are given in following theorem.
\begin{theorem} For a positive even integer $n$, \[A_4^{cf,GC,r}(n,d_H,\ell) = A_4^{cf,GC,rc}(n,d_H,\ell),\] and for a positive odd integer $n$, 
\begin{equation*}
    \begin{split}
        A_4^{cf,GC,r}(n,d_H+1,\ell) & \leq A_4^{cf,GC,rc}(n,d_H,\ell) \\ 
        & \leq A_4^{cf,GC,r}(n,d_H-1,\ell).
    \end{split}
\end{equation*}
\end{theorem}
\begin{proof}
Similar to the proof of \cite[Theorem 4.1]{doi:10.1089/10665270152530818}, the proof follows the fact that, for even $n$ and $\z_1$, $\z_2\in\Sigma^n$, $d_H(\z_1,\z_2)$ = $d_H(\z_1^c,\z_2^c)$ and, for odd $n$, $A_4^{cf,GC,r}(n,d_H+1,\ell) \leq A_4^{cf,GC,r}(n-1,d_H,\ell)$ and $A_4^{cf,GC,r}(n-1,d_H-1,\ell) \leq A_4^{cf,GC,r}(n-1,d_H,\ell)$.
\end{proof}
\begin{remark}
For positive integers $n$, $\ell$ $(<\lfloor n/2\rfloor)$ and $d_H$ ($\leq n$), 
\begin{equation*}
    A_4^{cf}(n,d_H,\ell) \geq A_4^{cf}(n,d_H,\ell+1). 
\end{equation*}
\end{remark}
\begin{remark} For positive integers $n$, $\ell$ $(<\lfloor n/2\rfloor)$ and $d_H$ ($\leq n$), 
\begin{equation*}
    \begin{split}
        A_4^{cf}(n,d_H,\ell) \geq A_4^{cf,GC}(n,d_H,\ell) & \geq A_4^{cf,GC,r}(n,d_H,\ell) \\
        & \geq A_4^{cf,GC,r,rc}(n,d_H,\ell).
    \end{split}
\end{equation*}
\end{remark}
\begin{remark} For positive integers $n$ and $d_H$ ($\leq n$), 
\begin{equation*}
\begin{split}
    & A_4^{cf,GC}(n,d_H,1) \geq A_4^{cf,GC}(n,d_H,\ell) \mbox{ and } \\
    & A_4^{GC,rc}(n,d_H) \geq A_4^{cf,GC,rc}(n,d_H,\ell),
\end{split}
\end{equation*}
where $A_4^{GC,rc}(n,d_H)$ is the  maximum  size  of a DNA  code  with $GC$ content and reverse-complement constraints.
\label{new lb remark}
\end{remark}

For  $n=1,2,\ldots, 10$ and $\ell=1,2,\ldots,\lfloor n/2\rfloor$, the lower bound of $A_4^{cf,GC,r,rc}(n,d_H,\ell)$ (maximum size of $\ell$ complete conflict free DNA code with Hamming, reverse, reverse-complement and $GC$ Content constraints) is calculated using the Construction \ref{lb of CCF code size cons}.

For $n=1,2,\ldots, 10$, $g$ = $\lfloor n/2\rfloor$ and $\ell=1,2,\ldots,\lfloor n/2\rfloor$, the maximum sizes for DNA codes with various parameters are listed in Table \ref{lb of CCF code size table}, where the DNA codes satisfy reverse, reverse-complement, $GC$ content constraints and each codeword of the DNA code is $\ell$ conflict free.
For given $n$ and $\ell$, first the seed set $\mathcal{S}$ is obtained such that each DNA string in the set contains $\lfloor n/2\rfloor$ $GC$ content. 
Further, for a random sub-set $R$ of the set $\mathcal{S}$, the $\ell$ conflict free DNA code $\mathscr{C}_{DNA}$ = $R\cup\{\x^r,\x^c:\x\in\mathscr{C}\}$ is obtained.  
From the seed set $\mathcal{S}$, the sub-set $R$ is generated $10^6$ times and, for each sub-set $R$, the size $|\mathscr{C}_{DNA}|$ ($\leq A_4^{cf,GC,r,rc}(n,d_H,\ell)$) is calculated. 
For given $n$ and $\ell$, the lower bound in the Table \ref{lb of CCF code size table} is the maximum size which is obtained among all the $\ell$ complete conflict free DNA codes. 

If a DNA string of length $n$ is $\lfloor n/2\rfloor$ conflict free then, from Definition \ref{conflict free def}, there does not exist two consecutive identical sub-strings of same length in the DNA string.
Therefore, the DNA string is free from consecutive repetition of DNA sub-string(s) of any length and such DNA string can be called complete conflict free. 
This motives the following definition.
\begin{definition}
A DNA code $\mathscr{C}(n,M,d)$ is called complete conflict free DNA code, if each codeword of the DNA code is complete conflict free, $i.e.$, $\lfloor n/2\rfloor$ conflict free.
\label{con. free code}
\end{definition}
For example, $\{ACGT,TGCA,ATCG,CGTA,GCAG\}$ is a complete conflict free DNA code $i.e.$, each codeword is $\lfloor 4/2\rfloor = 2$ conflict free.
The complete conflict free DNA code can be constructed using Construction \ref{lb of CCF code size cons} by taking $\ell$ = $\lfloor n/2\rfloor$.

Note that, for a positive integer $n$, if $\z\in S(n)$ then $\z^r,\z^c,\z^{rc}\in S(n)$, where $S(n)$ is the set of all complete conflict free DNA strings each of length $n$.
Also note, any DNA code $\mathscr{C}_{DNA}\subset S(n)$ is always be a complete conflict free DNA code. 

The lower bound of $A_4^{cf,GC,r,rc}(n,d_H,\lfloor n/2\rfloor)$ is enumerated in Table \ref{lb of CCF code size table} using Construction \ref{lb of CCF code size cons} for $n$ = $1,2,\ldots,10$ and Hamming distance $d_H$ = $1,2,\ldots,n$. 
Note that for any complete conflict free DNA code $\mathscr{C}_{DNA}(n,M,d_H)$, the code size $M\leq A_4^{cf,GC,r,rc}(n,d_H,\lfloor n/2\rfloor)$. 
 Also note that $A_4^{cf,GC}(n,1,\lfloor n/2\rfloor)$ = $A_4^{cf,GC,r,rc}(n,1,\lfloor n/2\rfloor)$ and $A_4^{cf}(n,1,\lfloor n/2\rfloor)$ = $A_4^{cf,r,rc}(n,1,\lfloor n/2\rfloor)$, for a positive integer $n$. 
Therefore, from the Remark \ref{conflict free comp} and Definition \ref{con. free code}, $A_4^{cf,GC,r,rc}(n,1,\lfloor n/2\rfloor)$ = $|S|$, where $|S|$ is enumerated at Step $4$ in Construction \ref{lb of CCF code size cons}.
Therefore, for $n$ = $1,2,\ldots,10$, the listed values $A_4^{cf,GC,r,rc}(n,1,\lfloor n/2\rfloor)$ are tight.
Similarly, from the definition of Hamming distance one can observe that, for any odd length DNA string $\x$, $d_H(\x,\x^r)< n$,  
\[
A_4^{cf,GC,r,rc}(n,n,\lfloor n/2\rfloor) = 
\left\{
\begin{array}{ll}
    2 & \mbox{if } n \mbox{ is odd}, \\
    4 & \mbox{if } n \mbox{ is even}.
\end{array}
\right.
\]
Hence, the values $A_4^{cf,GC,r,rc}(n,n,\lfloor n/2\rfloor)$ are also tight for $d_H$ = $n$ (= $1,2,3\ldots,10$) in the Table \ref{lb of CCF code size table}. 
In Table \ref{lb of CCF code size table}, values written in bold font indicate equal or improved values from \cite[Table I]{8424143} or \cite[Table II]{4418465}. 
For the equal and improved values, the corresponding code parameters have been considered and the respective DNA codes and their codewords are listed in Table \ref{codewords for improved codes}. 
Apart from the existing literature, all the specified constraints have been considered in Table \ref{lb of CCF code size table}. 
Some of those values are batter compared to the DNA code size with less constraints listed in Table \ref{optimal list}. 

Note that the values in Table I, \cite{8424143} are the lower bounds for maximum size of DNA codes satisfying $GC$ content constraint and free from homopolymers, and in Table II, \cite{4418465} are lower bounds for maximum size of DNA codes satisfying $GC$ content and reverse-complement constraints, on the other hand, values listed in Table \ref{lb of CCF code size table} in this paper are the lower bounds for maximum size of complete conflict free DNA codes satisfying Hamming, reverse, reverse-complement and $GC$ content constraints. 
From Remark \ref{new lb remark} and Table \ref{optimal list}, some new lower bounds are obtained for $A_4^{cf,GC}(n,d_H,1)$ and $A_4^{GC,rc}(n,d_H)$, where $A_4^{GC,rc}(n,d_H)$ is the maximum size of DNA code with codeword length $n$ and minimum distance $d_H$ which satisfies reverse-complement and $GC$ content constraints. 
For $A_4^{cf,GC}(n,d_H,1)$, the newly updated bounds are 
\[\begin{array}{ll}
    A_4^{cf,GC}(4,3,1)\geq 12, & A_4^{cf,GC}(6,4,1)\geq 20, \\
    A_4^{cf,GC}(8,6,1)\geq 12, & A_4^{cf,GC}(9,6,1)\geq 16, \\
    A_4^{cf,GC}(9,9,1)\geq 2, & A_4^{cf,GC}(10,7,1)\geq 16 \mbox{ and } \\
    A_4^{cf,GC}(10,8,1)\geq 8. & 
\end{array}\]
Similarly, for $A_4^{GC,rc}(n,d_H)$, the newly achieved bounds are 
\[
\begin{array}{ll}
    A_4^{GC,rc}(4,3)\geq 12, & A_4^{GC,rc}(6,4)\geq 20 \mbox{ and} \\ 
    A_4^{GC,rc}(9,9)\geq 2. & \\
     
\end{array}
\]

\begin{table}[t]
\caption{Codewords for complete conflict free DNA codes with Hamming, reverse, reverse-complement and $GC$ Content constraints. All the codes have code size improved from existing literature.}
\begin{center}\scriptsize
\begin{tabular}{|c|}
\hline
     Codewords for $(4,12,3)$ complete conflict free DNA code \\  \hline
     $ACTG$, $AGCT$, $ATGC$, $CAGT$, $CGTA$, $CTAG$, \\ 
     $GATC$, $GCAT$, $GTCA$, $TACG$, $TCGA$, $TGAC$ \\ \hline\hline
     Codewords for $(6,20,4)$ complete conflict free DNA code \\  \hline
     $ACAGTG$, $ACGTGA$, $AGCTAG$, $AGTGCA$, $ATCAGC$, \\
     $CACTGT$, $CATGTC$, $CGACTA$, $CTAGCT$, $CTGTAC$, \\
     $GACATG$, $GATCGA$, $GCTGAT$, $GTACAG$, $GTGACA$, \\
     $TAGTCG$, $TCACGT$, $TCGATC$, $TGCACT$, $TGTCAC$ \\ \hline\hline
     Codewords for $(8,12,6)$ complete conflict free DNA code \\  \hline
     $ACAGATCG$, $AGCTACTC$, $CATACGTC$, $CGATCTGT$, \\
     $CTCATCGA$, $CTGCATAC$, $GACGTATG$, $GAGTAGCT$, \\
     $GCTAGACA$, $GTATGCAG$, $TCGATGAG$, $TGTCTAGC$ \\ \hline\hline
     Codewords for $(9,16,6)$ complete conflict free DNA code \\  \hline
     $ACAGTAGCT$, $AGCTACTGT$, $AGTAGCATC$, $ATACAGACG$, \\ 
     $ATGATCGAG$, $CGTCTGTAT$, $CTACGATGA$, $CTCGATCAT$, \\ 
     $GAGCTAGTA$, $GATGCTACT$, $GCAGACATA$, $TACTAGCTC$, \\ 
     $TATGTCTGC$, $TCATCGTAG$, $TCGATGACA$, $TGTCATCGA$ \\ \hline\hline
     Codewords for $(9,2,9)$ complete conflict free DNA code \\  \hline
     $ACGATAGCA$, $TGCTATCGT$ \\ \hline\hline
     Codewords for $(10,16,7)$ complete conflict free DNA code \\  \hline
     $ACGTAGCAGA$, $ACTACAGACG$, $AGACGATGCA$, \\
     $AGCGACTATC$, $ATAGCTCGTG$, $CACGAGCTAT$, \\ 
     $CGTCTGTAGT$, $CTATCAGCGA$, $GATAGTCGCT$, \\
     $GCAGACATCA$, $GTGCTCGATA$, $TATCGAGCAC$, \\
     $TCGCTGATAG$, $TCTGCTACGT$, $TGATGTCTGC$, \\ 
     $TGCATCGTCT$ \\
     \hline\hline
     Codewords for $(10,8,8)$ complete conflict free DNA code \\  \hline
     $ACATGCGATC$, $CAGATACAGC$, $CGACATAGAC$,  \\ 
     $GATCGCATGT$, $GCTGTATCTG$, $GTCTATGTCG$,  \\ 
     $CTAGCGTACA$, $TGTACGCTAG$ \\ \hline
\end{tabular}
\end{center}
\label{codewords for improved codes}
\end{table}

\section{Mapping and their properties}
\label{sec:mapping}
In spite of the fact that the frequency of occurrence of slipped-strand mispairing errors in a complete conflict free DNA string is less then in a $\ell$ conflict free DNA string, the chances of occurrence of these errors is significantly low in a $\ell$ conflict free DNA string  for a sufficiently large $\ell$.
On the other hand, the computational complexity of Construction \ref{lb of CCF code size cons} is high.
Therefore, to sidestep the computational approach, a recursive mapping is defined algebraically in this section which ensures that the obtained DNA strings will be $\ell$ conflict free. 
Moreover, DNA codes satisfying all the constraints are also studied with respect to the mapping in this section.
\begin{definition}
For a positive integer $\ell$, consider $\x,\y\in\Sigma_{DNA}^\ell$ such that $\x\neq\y$. 
Define a map $f:\{\x,\x^c,\y,\y^c\}\times\{0,1\}\rightarrow\{\x,\x^c,\y,\y^c\}$ such that the Table \ref{Gen Map}($a$) holds.


\begin{table}[ht]\centering
	\caption{Mapping and an Example}
	\begin{tikzpicture}
		\draw (0,0) -- (4,0);
		\draw (0,0.1) -- (4,0.1);
		\node (a1) at (0.5,0.7) {\scriptsize{Binary}};
		\node (a2) at (0.5,0.3) {\scriptsize{Digit}};
		\draw (1,0.8) -- (1,-1.1);
		\draw (1.1,0.8) -- (1.1,-1.1);
		\node (a3) at (1.46,0.3) {$\x$};
		\draw (1.825,0.45) -- (1.825,-1.1);
		\node (a4) at (2.187,0.35) {$\x^c$};
		\draw (2.55,0.45) -- (2.55,-1.1);
		\node (a5) at (2.91,0.3) {$\y$};
		\draw (3.275,0.45) -- (3.275,-1.1);
		\node (a6) at (3.637,0.35) {$\y^c$};
		\node (a7) at (2.55,0.7) {\scriptsize{Previous  Nucleotide  Block}};
		
		\node (a8) at (0.5,-0.3) {$0$};
		\node (a9) at (1.46,-0.3) {$\y$};
		\node (a10) at (2.187,-0.25) {$\y^c$};
		\node (a11) at (2.91,-0.25) {$\x^c$};
		\node (a12) at (3.637,-0.3) {$\x$};
		\draw (0,-0.6) -- (4,-0.6);
		\node (a13) at (0.5,-0.9) {$1$};
		\node (a14) at (1.46,-0.85) {$\y^c$};
		\node (a15) at (2.187,-0.9) {$\y$};
		\node (a16) at (2.91,-0.9) {$\x$};
		\node (a17) at (3.637,-0.85) {$\x^c$};
		\node (a18) at (2,-1.4) {($a$) Mapping};

		\draw (4.5,0) -- (8.5,0);
		\draw (4.5,0.1) -- (8.5,0.1);
		\node (b1) at (5,0.7) {\scriptsize{Binary}};
		\node (b2) at (5,0.3) {\scriptsize{Digit}};
		\draw (5.5,0.8) -- (5.5,-1.1);
		\draw (5.6,0.8) -- (5.6,-1.1);
		\node (b3) at (5.96,0.3) {\scriptsize$CG$};
		\draw (6.325,0.45) -- (6.325,-1.1);
		\node (b4) at (6.687,0.3) {\scriptsize$GC$};
		\draw (7.05,0.45) -- (7.05,-1.1);
		\node (b5) at (7.41,0.3) {\scriptsize$AT$};
		\draw (7.775,0.45) -- (7.775,-1.1);
		\node (b6) at (8.137,0.3) {\scriptsize$TA$};
		\node (b7) at (7.05,0.7) {\scriptsize{Previous  Nucleotide  Block}};
		
		\node (b8) at (5,-0.3) {$0$};
		\node (b9) at (5.96,-0.3) {\scriptsize$AT$};
		\node (b10) at (6.687,-0.3) {\scriptsize$TA$};
		\node (b11) at (7.41,-0.3) {\scriptsize$GC$};
		\node (b12) at (8.137,-0.3) {\scriptsize$CG$};
		\draw (4.5,-0.6) -- (8.5,-0.6);
		\node (b13) at (5,-0.9) {$1$};
		\node (b14) at (5.96,-0.9) {\scriptsize$TA$};
		\node (b15) at (6.687,-0.9) {\scriptsize$AT$};
		\node (b16) at (7.41,-0.9) {\scriptsize$CG$};
		\node (b17) at (8.137,-0.9) {\scriptsize$GC$};
		\node (a18) at (6.5,-1.4) {(b) Example};
	\end{tikzpicture}
	\label{Gen Map}
\end{table}
\label{KG Map def}
\end{definition}

For example, one can obtain the mapping $f$ as given in Table \ref{Gen Map}(b) by considering $\x=CG$ and $\y=AT$ for $\ell$ = $2$.
One can read the table as $f(CG,0)$ = $AT$ and the rest follows.

\begin{encoding}
For positive integers $n$ and $\ell$, consider a mapping $f$ as defined in the Definition \ref{KG Map def}. 
A binary string $\a$ = $(a_1\ a_2 \ldots a_n)\in\{0,1\}^n$ is encoded into $\u$ = $(\u_1\ \u_2 \ldots \u_n)\in\{\x,\x^c,\y,\y^c\}^n$, in such a way that $\u_i=f(\u_{i-1},a_i)$ for each $i=2,3,\ldots,n$ and $\u_1=h(a_1)$, where $h:\{0,1\}\rightarrow\{\x,\x^c,\y,\y^c\}$ such that $h(0)=h(1)^c$.
Note that $\u_1\in\{\x,\x^c,\y,\y^c\}$ initiates the encoding of the binary string, and therefore the length of the encoded DNA string $\u$ is $n\ell$. 
Clearly, the encoding of two distinct binary strings are always distinct. 
Observe that for $i=2,3,\ldots,n$, $f(\u_{i-1},a_i)^c$ = $f(\u_{i-1}^c,a_i)$ = $f(\u_{i-1},\bar{a}_i)$, where $\bar{a}_i$ is binary complement of $a_i$.
\label{gen encod}
\end{encoding}

For positive integers $n$ and $\ell$ $(<n)$, consider a subset $S\subseteq\{0,1\}^n$. 
Each binary string from the set $S$ is encoded using the Non-Homopolymer map, and the set of all encoded DNA strings is denoted by $f(S)$. 
In particular, let $f(\{0,1\}^n)\subseteq\Sigma_{DNA}^{n\ell}$ be denoted by $\mathcal{B}$ the set of all possible DNA strings of length $n\ell$ and obtained by the encoding using Non-Homopolymer map.
The set $\mathcal{B}$ is given by Equation (\ref{equ b}).

\begin{figure*}
\begin{equation}
\mathcal{B} = \left\{
\begin{array}{ll}
    \{\u_1,\u_1^c\} & n=1, \\
    \{\u_1,\u_1^c\}\times(\{\u_2,\u_2^c\}\times\{\u_1,\u_1^c\})^{(n-1)/2} & n\mbox{ is odd integer and } n > 1, \\
    (\{\u_1,\u_1^c\}\times\{\u_2,\u_2^c\})^{n/2} & n\mbox{ is positive even integer}.
\end{array}
\right.
\label{equ b}
\end{equation}
\end{figure*}

For example, consider the mapping as given in Table \ref{Gen Map}(b) and $\u_1=\x$.
The binary string $(0\ 1\ 1)$ is encoded into the DNA string $(\x\ \y^c\ \x^c)$ = $CGTAGC$. 

The following Theorem 
ensures that the encoded DNA strings obtained from the mapping will be conflict free.
\begin{theorem}
For a positive integers $\ell$, if $\x,\y\in\Sigma_{DNA}^\ell$ such that each DNA string in the set $\{(\x\ \y^*\ \x^*\ \y^*), (\y\ \x^*\ \y^*\ \x^*): \x^*\in\{\x,\x^c\} \mbox{ and } \y^*\in\{\y,\y^c\}\}$ is $2\ell-1$ conflict free then any binary string will be encoded into a $2\ell-1$ conflict free DNA string using Encoding \ref{gen encod}.
\label{conflict free theorem}
\end{theorem}
\begin{proof}
From the Remark \ref{conflict free comp}, if any DNA string in the set $\{(\x\ \y^*\ \x^*\ \y^*), (\y\ \x^*\ \y^*\ \x^*): \x^*\in\{\x,\x^c\} \mbox{ and } \y^*\in\{\y,\y^c\}\}$ is $2\ell-1$ conflict free then all the DNA sub-strings $(\x^*\ \y^*\ \x^*\ \y^*)$ and $(\y^*\ \x^*\ \y^*\ \x^*)$ are also $2\ell-1$ conflict free. 
Therefore, in the encoded DNA string $\u$ = $(\u_1\ \u_2 \ldots \u_n)$ (using Encoding \ref{gen encod}), for $1\leq i\leq n-3$, consider $(\u_i\ \u_{i+1}\ \u_{i+2}\ \u_{i+3})$, which is also $2\ell-1$ conflict free, where $\u_i\in\{\x,\x^c,\y,\y^c\}$ and the proof follows.
\end{proof}
\begin{proposition}
For a positive integers $\ell$, if $\x,\y\in\Sigma_{DNA}^\ell$ such that each of the DNA strings $(\x\ \y)$, $(\x\ \y^c)$, $(\y\ \x)$ and $(\y\ \x^c)$ is $\ell$ conflict free then any binary string will be encoded into a $\ell$ conflict free DNA string using Encoding \ref{gen encod}.
\label{conflict free proposition}
\end{proposition}
\begin{theorem}
For a positive integers $\ell$, if $\x,\y\in\Sigma_{DNA}^\ell$ such that each of the DNA strings $(\x\ \y\ \x)$, $(\x\ \y\ \x^c)$, $(\x\ \y^c\ \x)$ and $(\x\ \y^c\ \x^c)$ is free from reverse-complement sub-string(s) then the encoded DNA string using Encoding \ref{gen encod} is free from hairpin like secondary structure. 
\label{hair pin theorem}
\end{theorem}
\begin{proof}
From Remark \ref{kg hairpin remark}, if $(\x\ \y\ \x)$ is free from reverse-complement sub-string(s) then all of the corresponding strings $(\x^c\ \y^c\ \x^c)$, $(\y\ \x\ \y)$ and $(\y^c\ \x^c\ \y^c)$ will be free from reverse-complement sub-string(s).
Similarly remaining all $12$ triplets are also free from reverse-complement sub-string(s).
Hence, from Definition \ref{kg hairpin definition} and Encoding \ref{gen encod}, the encoded DNA string obtained from any binary string will be free from hairpin like secondary structure.
\end{proof}
The Theorem \ref{conflict free dna string} imposes a condition on binary strings such that the encoded DNA strings will be complete conflict free.
\begin{theorem}
For positive integers $n$ and $\ell$, consider $\x,\y\in\Sigma_{DNA}^\ell$ such that the DNA strings $(\x\ \y)$, $(\x\ \y^c)$, $(\y\ \x)$ and $(\y\ \x^c)$ are $\ell$ conflict free.
If $\a$ = $(a_1\ a_2\ldots a_n)$ is a binary string such that $2\mu<\sum_{i=\lambda+1}^{\lambda+2\mu} (a_ia_{2\mu+i}+\bar{a}_i\bar{a}_{2\mu+i})$ for each positive even integer $2\mu$ from the set $\{1,2,\ldots,\lfloor n/2\rfloor\}$ and $\lambda=0,1,\ldots, n-2\mu$ then the binary string $\a$ will be encoded (using Encoding \ref{gen encod}) into a complete conflict free DNA string of length $n\ell$.
\label{conflict free dna string}
\end{theorem}
\begin{proof}
Consider a binary string $\a$ = $(a_1\ a_2\ldots a_n)$ which is encoded into DNA string $\u$ = $(u_1\ u_2\ldots u_n)$ using Encoding \ref{gen encod}.
For each positive even integer $2\mu$ from the set $\{1,2,\ldots,\lfloor n/2\rfloor\}$, the DNA block $u_{2\mu+i}\in\{u_i,u_i^c\}$.
For any binary symbol $a_i$, $a_{2\mu+i}\in\{0,1\}$, 
\[
(a_ia_{2\mu+i}+\bar{a}_i\bar{a}_{2\mu+i})= 
\left\{
\begin{array}{ll}
    1 & \mbox{if } \mbox{ $a_i=a_{2\mu+i}$} \\
    0 & \mbox{otherwise. }
\end{array}
\right.
\]
Therefore, $\sum_{i=1}^{2\mu} (a_ia_{2\mu+i}+\bar{a}_i\bar{a}_{2\mu+i})=2\mu$ if and only if $a_i=a_{2\mu+i}$ for each $i=1,2,\ldots,2\mu$. 
If the origin is shifted with $\lambda$, and $2\mu<\sum_{i=\lambda+1}^{\lambda+2\mu} (a_ia_{2\mu+i}+\bar{a}_i\bar{a}_{2\mu+i})$ for each $\lambda$ and $\mu$, then from Encoding \ref{gen encod} and Proposition \ref{conflict free proposition}, the encoded DNA string will be a complete conflict free.
\end{proof}
The $GC$ content of encoded DNA string is calculated in the following lemma.
\begin{lemma}
For positive integers $n$ and $\ell$, consider $\x,\y\in\Sigma_{DNA}^\ell$ with $GC$ content $g_\x$ and $g_\y$. 
For the encoded DNA string $\u$ = $(\u_1\ \u_2 \ldots \u_n)\in\{\x,\x^c,\y,\y^c\}^n$ using Encoding \ref{gen encod}, the $GC$ content of $\u$ will be 
\begin{equation*}
    g_\u = 
    \left\{
    \begin{array}{ll}
        g_\x+(g_\x+g_\y)(n-1)/2 & \mbox{ if } n \mbox{ is odd integer} \\
        (g_\x+g_y)n/2 & \mbox{ if } n \mbox{ is even integer,}
    \end{array}
    \right.
\end{equation*}
where $\u_1\in\{\x,\x^c\}$.
\label{kg gen gc cont}
\end{lemma}
\begin{proof}
For positive integers $n$ and $\ell$ ($<n$), let a binary string $\a$ = $(a_1\ a_2 \ldots a_n)\in\{0,1\}^n$ be encoded into some $\u$ = $(\u_1\ \u_2 \ldots \u_n)\in\{\x,\x^c,\y,\y^c\}^n$ using Encoding \ref{gen encod}.
In the Encoding \ref{gen encod}, if $\u_1\in\{\x,\x^c\}$ then DNA blocks $\u_{2j}\in\{\y,\y^c\}$ and $\u_{2j+1}\in\{\x,\x^c\}$, for $1\leq2j,2j+1\leq n$. 
Since, the $GC$ content of a DNA string and its complement DNA string are the same, the $GC$ content of each sub-string ($\u_{2j}\ \u_{2j+1}$) is $g_\x+g_\y$.
Hence, if $n$ is even, the $GC$ content of the encoded DNA string $\u$ is $g_\u$ = $(g_\x+g_y)n/2$ and if $n$ is odd then $g_\u$ = $g_\x+(g_\x+g_y)(n-1)/2$.
\end{proof}
\begin{remark}
Note that in Lemma \ref{kg gen gc cont}, if $\u_1\in\{\y,\y^c\}$ then 
\begin{equation*}
    g_\u = 
    \left\{
    \begin{array}{ll}
        g_\y+(g_\x+g_\y)(n-1)/2 & \mbox{ if } n \mbox{ is odd integer} \\
        (g_\x+g_y)n/2 & \mbox{ if } n \mbox{ is even integer}.
    \end{array}
    \right.
\end{equation*}
\label{kg gen gc cont remark}
\end{remark}
The following Theorem ensures that the $GC$ content of the encoded DNA string (using Encoding \ref{gen encod}) is almost $50\%$ of the length.
\begin{theorem}
For a positive integer $n$ and $\ell$, let $\x,\y\in\Sigma_{DNA}^\ell$.
If the $GC$ content of $\x$ and $\y$ are $\lfloor\ell/2\rfloor$ and $\lceil\ell/2\rceil$ respectively, then the $GC$ content of any encoded DNA string $\u$ (using Encoding \ref{gen encod}) of length $n\ell$ is 
\begin{equation*}
    g_\u = 
    \left\{
    \begin{array}{ll}
        \lfloor n\ell/2\rfloor & \mbox{ if } \u_1\in\{\x,\x^c\}, \mbox{ and } \\
        \lceil n\ell/2 \rceil & \mbox{ if } \u_1\in\{\y,\y^c\}.
    \end{array}\right.
\end{equation*}
\label{KG GC cont}
\end{theorem}
\begin{proof}
The Theorem follows from Lemma \ref{kg gen gc cont} and Remark \ref{kg gen gc cont remark}, where $g_\x$ = $\lfloor\ell/2\rfloor$ and $g_\y$ = $\lceil\ell/2\rceil$.
\end{proof}

\begin{theorem}
For positive integers $n$ and $\ell$, let $\x,\y\in\Sigma_{DNA}^\ell$.
Using Encoding \ref{gen encod}, if a binary string $(a_1\ a_2 \ldots a_n)$ is encoded into some $\u\in\{\x,\x^c,\y,\y^c\}^n$ then the binary string $(\bar{a}_1\ a_2 \ldots a_n)$ is encoded into $\u^c$, where $\bar{a}_1$ is the binary complement of $a_1$. 
 \label{kg comp special case}
\end{theorem}
\begin{proof}
The proof is done using induction on the index $i$ ($i=1,2,\ldots,n$).
Now $f(\z,0)^c$ = $f(\z^c,0)$ and $f(\z,1)^c$ = $f(\z^c,1)$, for each $\z\in\{\x,\x^c,\y,\y^c\}$, from Definition \ref{KG Map def}. 
Consider the binary strings $(a_1\ a_2\ a_3\ldots a_n)$ and $(\bar{a}_1\ a_2\ a_3\ldots a_n)$, that are encoded into some DNA strings $(u_1\ u_2\ u_3\ldots u_n)$ and $(v_1\ v_2\ v_3\ldots v_n)$. 
From Encoding \ref{gen encod}, $u_1^c$ = $h(a_1)^c$ = $h(\bar{a}_1)$ = $v_1$. 
Let $u_i^c$ = $v_i$, for some $i\in\{1,2,\ldots,n\}$.
Consider $u_{i+1}^c$ = $f(u_i,a_{i+1})^c$ = $f(u^c_i,a_{i+1})$ = $f(v_i,a_{i+1})$ = $v_{i+1}$.
Therefore, from induction, the binary strings $(a_1\ a_2\ a_3\ldots a_n)$ and $(\bar{a}_1\ a_2\ a_3\ldots a_n)$ are encoded into DNA strings which are complement to each other.
Similarly, the analogous statement for reverse can be proved using induction.
\end{proof}
In the following two theorems, 
the hamming distance between two DNA strings is calculated for binary strings with hamming distance $1$ and $2$.
\begin{theorem}
For positive integers $n$ and $\ell$, 
consider the binary strings $\a$ = $(a_1\ a_2\ldots a_n)$ and $\b$ = $(a_1\ a_2\ldots a_{i-1}\ \bar{a}_i\ a_{i+1}\ldots a_n)$ ($1\leq i\leq n$), that are encoded into DNA strings $\u$ = $(\u_1\ \u_2\ldots\u_n)$ and $\v$ = $(\v_1\ \v_2\ldots\v_n)$, where $\bar{a}_i$ is binary complement of $a_i$. 
Than, $d_H(\u,\v)$ = $\ell(n-i+1)$.
\label{kg theorem comp. spe. case}
\end{theorem}
\begin{proof}
Consider the binary strings $\a$ = $(a_1\ a_2 \ldots a_{i-1}\ a_i\ a_{i+1}\ldots a_n)$ and $\b$ = $(a_1\ a_2 \ldots a_{i-1}\ \bar{a}_i\ a_{i+1}\ldots a_n)$ which are encoded into $\u$ = $(\u_1\ \u_2 \ldots \u_{i-1}\ \u_i\ \u_{i+1}\ldots \u_n)$ and $\v$ = $(\v_1\ \v_2 \ldots \v_{i-1}\ \v_i\ \v_{i+1}\ldots \v_n)$ respectively.
From Encoding \ref{gen encod}, $\v_j$ = $\u_j$ ($j=1,2,\ldots i-1$) and $\v_j$ = $\u_j^c$ ($j=i,i+1\ldots,n$).
Therefore, $d_H(\u,\v)$ = $\ell(n-i+1)$, since $d_H(\x,\x^c)$ = $d_H(\y,\y^c)$ = $\ell$.
\end{proof}

\begin{theorem}
For positive integers $n$ and $\ell$, 
consider the binary strings $\a$ = $(a_1\ a_2\ldots a_n)$ and $\b$ = $(a_1\ a_2\ldots a_{i-1}\ \bar{a}_i\ a_{i+1}\ldots a_{j-1}\ \bar{a}_j\ a_{j+1}\ldots a_n)$ ($1\leq i<j\leq n$) that are encoded into DNA strings $\u$ = $(\u_1\ \u_2\ldots\u_n)$ and $\v$ = $(\v_1\ \v_2\ldots\v_n)$, where $\bar{a}_i$ is binary complement of $a_i$.
Than, $d_H(\u,\v)$ = $\ell(j-i)$.
\label{kg hamming 2 theorem}
\end{theorem}
\begin{proof}
The proof is similar to the theorem and follows from the fact that for any DNA string $\x$, $(\x^c)^c= \x$. 
\end{proof}
The following theorem provides a bound on between Hamming distance on binary strings and the encoded DNA strings.
\begin{theorem}
For positive integers $n$, $\ell$ and $\sigma$ $(\leq \ell)$, consider $\x,\y\in\Sigma_{DNA}^\ell$ and 
$\sigma$ = $\min\{d_H(\z_1,\z_2),n-d_H(\z_1,\z_2):\z_1\in\{\x,\x^c\}\mbox{ and }\z_2\in\{\y,\y^c\} \}$.
Let binary strings $\a$, $\b\in\{0,1\}^n$ be encoded into the DNA strings $\u$ = $(\u_1\ \u_2\ldots\u_n)$ and $\v$ = $(\v_1\ \v_2\ldots\v_n)$.
For some $a,b\in\{0,1\}$, if $\u'$ = $(\u\ f(\u_n,a))$ and $\v'$ = $(\v\ f(\v_n,b))$ then 
$d_H(\u',\v')\leq$ 
\[
\left\{
\begin{array}{ll}
    m(d_H(\a,\b)+d_H(a,b)) & \mbox{ if }d_H(\a,\b)\mbox{ is even} \\
    m(d_H(\a,\b)+|1-d_H(a,b)|) &\mbox{ if } d_H(\a,\b)\mbox{ is odd}.
\end{array}
\right.
\]
\label{kg isometry cases}
\end{theorem}
\begin{proof}
The proof follows from the following two facts. 
(1) For any $\x,\y\in\Sigma_{DNA}^\ell$, if $d_H(\x,\y)$ = $t$ then $d_H(\x^c,\y)\geq \ell-t$; and
(2) For $\z\in\{\x,\x^c,\y,\y^c\}$ and $a,b\in\{0,1\}$, $f(\z,c)^c$ = $f(\z^c,c)$ = $f(\z,\bar{c})$ [Definition \ref{KG Map def}].
So we can derive, $d_H(\u',\v')\leq$  
\[ 
\left\{
\begin{array}{ll}
    \sigma(d_H(\a,\b)+1) & \mbox{ if }d_H(\a,\b)\mbox{ is even and } a\neq b\\
    \sigma d_H(\a,\b) & \mbox{ if }d_H(\a,\b)\mbox{ is even and } a= b\\
    \sigma d_H(\a,\b) &\mbox{ if } d_H(\a,\b)\mbox{ is odd and } a\neq b \\
    \sigma(d_H(\a,\b)+1) &\mbox{ if } d_H(\a,\b)\mbox{ is odd and } a = b.
\end{array}
\right.
\]
Hence the proof follows.
\end{proof}
In order to establish the proposed mapping as an isometry from the set of binary strings to set of DNA strings, where Hamming distance is taken for the set of DNA strings, we introduce a new distance between two binary strings in the following definition.   
\begin{definition}
For a positive integer $n$ and alphabet set $\Sigma$ of size $q$, let $\a$ = $(a_1\ a_2\ldots a_n)\in\Sigma^n$ and $\b$ = $(b_1\ b_2\ldots b_n)\in\Sigma^n$ be two vectors of length $n$. 
For a set $P$ = $\{i:a_i\neq b_i, 1\leq i\leq n\mbox{ and }a_i,b_i\in\Sigma\}$, consider the support set $S$ such that 
\begin{equation*}
    S=
    \left\{
    \begin{array}{ll}
        P & \mbox{ if } d_H(\a,\b)\mbox{ is even and } n>1 \\
        P\cup\{n+1\} & \mbox{ if } d_H(\a,\b)\mbox{ is odd and } n>1,
    \end{array}
    \right.
\end{equation*}
such that, for each $s_j\in S$ $(j=1,2,\ldots,|S|-1)$, $s_j<s_{j+1}$. 
We define a map $d:\Sigma^n\times\Sigma^n\rightarrow\R$ such that 
\begin{equation*}
    d(\a,\b) = 
    \left\{
    \begin{array}{ll}
    \ell\sum_{i=1}^{|S|/2} (s_{2i}-s_{2i-1}) & \mbox{ if } |S|>0,\\
    0 & \mbox{ if } |S|=0, 
    \end{array}
    \right.
    \label{kg distance}
\end{equation*}
where $\ell$ is a positive integer. 

\begin{remark}
The map $d:\Sigma^n\times\Sigma^n\rightarrow\R$ is a distance.
Note that, for a code $\mathscr{C}\subseteq\Sigma^n$, the minimum distance $d$ = $\min\{d(\a,\b):\a\neq\b,\mbox{ and }\a,\b\in\mathscr{C}\}$.

\end{remark}

\end{definition}
For example, consider $n=5$, $\ell=2$ and $\Sigma=\{0,1\}$. For $\a=(1\ 1\ 1\ 1\ 0)$ and $\b=(0\ 1\ 1\ 0\ 0)$, $d(\a,\b)=6$, where $S=\{1,4\}$. 

\begin{theorem}
For positive integers $n$ and $\ell$, 
the Encoding \ref{gen encod} is a distance preserving encoding between $(\{0,1\}^n,d)$ and $(\mathcal{B},d_H)$.
\label{distance preserving theorem}
\end{theorem}
\begin{proof}
The theorem is proved using induction on the string length $n$. 
The base case, $n=1$, is obvious from Definition \ref{kg distance}. 
For the inductive step, assume that the distance is preserved for $n=k$, where $\a$ = $(a_1\ a_2\ldots a_k)$ and $\b$ = $(b_1\ b_2\ldots b_k)$ with the support set $S$, are encoded into DNA strings $\u$ = $(\u_1\ \u_2\ldots \u_k)$ and $\v$ = $(\v_1\ \v_2\ldots \v_k)$. 
To prove that the distance is preserved for $n=k+1$, consider $\a'$ = $(\a\ a_{k+1})$ = $(a_1\ a_2\ldots a_k\ a_{k+1})$ and $\b'$ = $(\b\ b_{k+1})$ = $(b_1\ b_2\ldots b_k\ b_{k+1})$ with the support set $S'$, where $a_{k+1},b_{k+1}\in\{0,1\}$.
Let the strings $\a$ and $\b$ be encoded into DNA strings $\u'$ = $(\u_1\ \u_2\ldots \u_k\ \u_{k+1})$ and $\v'$ = $(\v_1\ \v_2\ldots \v_k\ \v_{k+1})$, where $\u_{k+1},\v_{k+1}\in\{\x,\y,\x^c,\y^c\}$. 
Now, there are four cases.
(i) Consider $d_H(\a,\b)$ is even and $a_{k+1}$ = $b_{k+1}$. 
Note that both $\u_k$ and $\v_k$ are member of either $\{\x,\x^c\}$ or $\{\y,\y^c\}$. 
On the other hand, $d(\a',\b')$ = $d(\a,\b)$, since $S'$ = $S$.
(ii) If $d_H(\a,\b)$ is even and $a_{k+1}\neq b_{k+1}$ then $\u_{k+1}$ = $\v_{k+1}^c$.
So, for $\ell$ = $d_H(\x,\x^c)$ = $d_H(\y,\y^c)$, $d_H(\u',\v')$ = $d_H(\u,\v)+\ell$ and $d(\a',\b')$ = $d(\a,\b)+\ell$ since, $S'$ = $\{k+1,k+2\}\cup S$.
(iii) If $d_H(\a,\b)$ is odd and $a_{k+1}$ = $b_{k+1}$ then $d_H(\u',\v')$ = $d_H(\u,\v)+m$ and $d(\a',\b')$ = $d(\a,\b)+\ell$ since, $S'$ = $\{k+2\}\cup S\backslash\{k+1\}$.
(iv) If $d_H(\a,\b)$ is odd and $a_{k+1}\neq b_{k+1}$ then $d_H(\u',\v')$ = $d_H(\u,\v)$ and $d(\a',\b')$ = $d(\a,\b)$ since, $S'$ = $S$.
Note that all the four cases follows from the Theorem \ref{kg isometry cases}.
Hence the result follows.
\end{proof}

For positive integers $n$ and $\ell$ $(<n)$, consider a subset $\mathscr{C}\subseteq\{0,1\}^n$. 
Each binary string from the set $\mathscr{C}$ is encoded using the Encoding \ref{gen encod}, and the set of all encoded DNA strings is denoted by $f(\mathscr{C})$. 

\begin{theorem}
For a binary code $\mathscr{C} (n, M, d)$ where $d$ is considered as the distance defined on Definition \ref{kg distance},
a DNA code $f(\mathscr{C})$ with codeword length $n\ell$, size $M$ and minimum Hamming distance $d_H=d$ can be constructed using Encoding \ref{gen encod}. 
\end{theorem}
\begin{proof}
The proof follows from Encoding \ref{gen encod} and Theorem \ref{distance preserving theorem}.
\end{proof}

\begin{theorem}
For any binary code $\mathscr{C}$ with minimum distance $d\leq n\ell/2$ ($n,\ell\in\Z^+$), there exists a DNA code $f(\mathscr{C}) (n\ell,M,d_{H})$ with complement constraint. 
\label{comp dis}
\end{theorem}
\begin{proof}
Let binary strings $\a$ and $\b$ of length $n$ are encoded into DNA strings $\u$ and $\v$ of length $n\ell$.
Form the property of the complement of a DNA string, $d_H(\u,\v^c)\geq n\ell-d_H(\u,\v)$.
From Theorem \ref{distance preserving theorem}, $d_H(\u,\v)$ = $d(\a,\b)\leq n\ell/2$.
Therefore, $d_H(\u,\v^c)\geq n\ell/2$ and hence the proof follows.
\end{proof}

\begin{theorem}
For a positive integer $n$, if a binary linear code with codeword length $n$ contains $(1\ 0\ 0\ldots 0)$ as a codeword then the encoded DNA code (using Encoding \ref{gen encod}) will satisfy complement constraint. 
\end{theorem}
\begin{proof}
Consider a binary linear code containing the codeword $(1\ 0\ 0\ldots 0)$ of length $n$.
For any codeword $\a$ = $(a_1\ a_2\ldots a_n)$ of the binary linear code, $(1\ 0\ 0\ldots 0) + (a_1\ a_2\ldots a_n)$ = $(\bar{a}_1\ a_2\ldots a_n)$ = $\b$ is also a codeword of the code. 
Therefore, from Theorem \ref{kg comp special case}, for each binary codeword $\a$, there exists a binary codeword $\b$ such that the encoded DNA strings from $\a$ and $\b$ will be complement to each other. 
Hence, by the distance property, the theorem is proved.
\end{proof}

\begin{theorem}
For positive integers $n$ and $\ell$, let $\x,\y\in\Sigma_{DNA}^\ell$.
Than, for any $\u,\v\in f(\{0,1\}^n)$, $d_H(\u,\v^r)\geq$ $\min\{d_H(\x,\x^r),d_H(\y,\y^r),d_H(\x,\x^{rc}),d_H(\y,\y^{rc})\}$ for odd $n$, and $d_H(\u,\v^r)\geq$ $n\min\{d_H(\x,\y^r),d_H(\x,\y^{rc})\}$ for even $n$. 
\label{kg reverse condition}
\end{theorem}
\begin{proof}
For $\x,\y\in\Sigma_{DNA}^\ell$, let the binary strings $\a,\b\in \{0,1\}^n$ of length $n$ be encoded into DNA strings $\u=(\u_1\ \u_2\ldots \u_n)$, $\v = (\v_1\ \v_2\ldots\v_n)$ in $\{\x,\x^c,\y,\y^c\}$, 
where $\u_{2i},\v_{2i}\in\{\u_2,\u_2^c\}$ and $\u_{2i+1},\v_{2i+1}\in\{\u_1,\u_1^c\}$ for $1\leq 2i,2i+1\leq n$. 
The set $f(\{0,1\}^n)$ is the collection of all possible DNA strings such that obtained DNA blocks will be from $\{\u_2,\u_2^c\}$ and $\{\u_1,\u_1^c\}$ at even positions and odd positions respectively.
Consider $d_H(\u,\v^r)=\sum_{j=1}^nd_H(\u_j,\v^r_{n-j+1})$. 
Now two cases may arise.

Case 1: If $n$ is odd then $j$ and $n-j+1$ both are either even or odd. 
If both $j$ and $n-j+1$ are even then $\u_j,\v_{n-j+1}\in\{\u_2,\u_2^c\}$, and if both $j$ and $n-j+1$ are odd then $\u_j,\v_{n-j+1}\in\{\u_1,\u_1^c\}$. 
Therefore, $\u,\v^r\in f(\{0,1\}^n)$ and, from Encoding \ref{gen encod}, $d_H(\u_j,\v^r_{n-j+1})\geq\min\{d_H(\x,\x^r),d_H(\y,\y^r),d_H(\x,\x^{rc}),d_H(\y,\y^{rc})\}$. 

Case 2: If $n$ is even then the parity $j$ and $n-j+1$ will be different. 
So, for even $j$, $\u_j\in\{\u_2,\u_2^c\}$ and $\v_{n-j+1}\in\{\u_1,\u_1^c\}$, and, for odd $j$, $\u_j\in\{\u_1,\u_1^c\}$ and $\v_{n-j+1}\in\{\u_2,\u_2^c\}$. 
Therefore, from Encoding \ref{gen encod} and the fact that, for any $\z_1,\z_2\in\Sigma_{DNA}^\ell$, $d_H(\z_1,\z_2^r)$ = $d_H(\z_1^r,\z_2)$, we obtain $d_H(\u_j,\v_{n-j+1}^r)\geq\min\{d_H(\x,\y^r),d_H(\x,\y^{rc})\}$.
Hence the result follows for every $n$.
\end{proof}

\begin{theorem}
For an even positive integer $n$ and a positive integer $\ell$, consider $\x,\y\in\Sigma_{DNA}^\ell$ such that $d_H(\x,\y^{rc})$ = $d_H(\x,\y^r)=\ell$. 
Then, the DNA codes constructed using Encoding \ref{gen encod} will satisfy the reverse constraint. 
\label{kg reverse theorem}
\end{theorem}
\begin{proof}
If $d_H(\x,\y^{rc})$ = $d_H(\x,\y^r)=\ell$ then, from Theorem \ref{kg reverse condition}, $d_H(\u,\v^r)\geq n\min\{d_H(\x,\y^r),d_H(\x,\y^{rc})\}=n\ell$. 
But the length of the encoded DNA string is $n\ell$ so, $d_H\leq n\ell$ and therefore, $d_H(\u,\v^r)\geq d_H$ for any DNA code constructed using Encoding \ref{gen encod}.
\end{proof}

\begin{lemma}
For positive integers $n$ and $\ell$, 
if the binary strings $\a$ and $\b$ of length $n$ are encoded into DNA strings $\u$ and $\v$ using Encoding \ref{gen encod} then 
\[
n-\lfloor d_H(\a,\b)/2\rfloor\geq \frac{1}{\ell}d_H(\u,\v)\geq\lceil d_H(\a,\b)/2\rceil.
\]
\label{kg half dis lemma}
\end{lemma}
\begin{proof}
For a positive integer $n$, if $S\subseteq\{1,2,\ldots,n,n+1\}$ is a set with even cardinality such that, for each $s_j\in S$ $(j=1,2,\ldots,|S|-1)$, $s_j<s_{j+1}$ then one can observe that 
\[
n-\frac{|S|}{2}\geq\sum_{i=1}^{|S|/2} (s_{2i}-s_{2i-1})\geq\frac{|S|}{2}.
\]
From the Definition \ref{kg distance}, $d_H(\a,\b)\in\{|S|,|S|-1\}$, and therefore, the proof follows.
\end{proof}

\begin{theorem}
For positive integers $n$ and $\ell$, consider a binary code with minimum hamming distance $d_H$ and minimum distance $d$.
Than $n-\lfloor d_H/2\rfloor\geq \frac{d}{\ell}\geq\lceil d_H/2\rceil$.
\label{kg half dis}
\end{theorem}
\begin{proof}
For any code $\mathscr{C}$ with minimum Hamming distance $d_H$, if $\a,\b\in\mathscr{C}$ then $\lceil d_H(\a,\b)/2\rceil\geq\lceil d_H/2\rceil$ and $n-\lfloor d_H/2\rfloor \geq n-\lfloor d_H(\a,\b)/2\rfloor\geq\lceil d_H/2\rceil$. 
The proof follows from Lemma \ref{kg half dis lemma} and Theorem \ref{distance preserving theorem}.
\end{proof}
In the following theorem, a constraint on binary string is imposed in such a way that the encoded DNA string will be complete conflict free.
\begin{theorem}
For positive integers $n$, $\ell$ and any positive even integer $2\mu\in\{1,2,\ldots,\lfloor n/2\rfloor\}$, consider a binary code with codeword length $n$, such that for each codeword $(a_1\ a_2\ldots a_n)$, $2\mu<\sum_{i=\lambda+1}^{\lambda+2\mu} (a_ia_{2\mu+i}+\bar{a}_i\bar{a}_{2\mu+i})$, where $\lambda=1,2,\ldots, n-2\mu$.
Then there exists a complete conflict free DNA code with codeword length $n\ell$.
\end{theorem}
\begin{proof}
The proof follows from Definition \ref{con. free code} and Theorem \ref{conflict free dna string}. 
\end{proof}

\section{Conflict Free DNA codes}
\label{sec:family}
For each positive integer $\ell$ = $2,3\ldots7$, all different possibilities of ($\x$, $\y$) are computed in Table \ref{hairpin l table}.
Considering any pair ($\x,\y$) from Table \ref{hairpin l table}, the obtained DNA code will satisfy the Hamming, reverse, reverse-complement and $\lfloor n\ell/2\rfloor$-$GC$ content constraints.
Moreover, each DNA codeword obtained from the encoding will be $2\ell-1$ conflict free. 

For $\x$ = $ACT$ and $\y$ = $CTG$, the DNA code obtained from $[7,4,3]$ binary Hamming code is given in Table \ref{Hanning codewords table}. 

For various binary codes, the parameters of encoded DNA codes are listed in Table \ref{Code comparison table}. 
\begin{table}[t]
\caption{Parameters for conflict free DNA codes encoded from binary codes.}
\begin{center}
\begin{tabular}{|c|m{0.6cm}|m{0.6cm}|m{0.6cm}|}
\hline
Binary Code & \multicolumn{3}{c|}{DNA Code Parameters} \\ \cline{2-4}
(Encoded from) & $n\ell$ & $M$ & $d_H$  \\ \hline \hline
\multicolumn{1}{|l|}{[5,2,5] Repetition Code} & $5\ell$ & $2$ & $3\ell$  \\  \hline
\multicolumn{1}{|l|}{[7,4,3] Hamming Code} & $7\ell$ & $16$ & $2\ell$  \\  \hline
\multicolumn{1}{|l|}{[8,4,2] Reed Muller Codes} & $8\ell$ & $256$ & $\ell$   \\ \hline
\multicolumn{1}{|l|}{(15, 256, 5) Nordstrom-Robinson Code} & $15\ell$ & $256$ & $3\ell$  \\ \hline
\multicolumn{1}{|l|}{[23,12,7] Golay Code} & $23\ell$ & $4096$ & $4\ell$  \\ \hline
\end{tabular}
\end{center}
\label{Code comparison table}
\end{table}
\subsection{Reed-Muller code:} The binary Reed Muller codes are introduced by Reed and Muller in 1954 \cite{6499441}. 
For two non-negative positive integers $m$ and $r$ ($r\leq m$), the $r^{th}$ order binary Reed Muller code $\mathcal{R}(r,m)$ is a linear code of length $2^m$, code size $2^{\sum_{i=1}^r\binom{m}{i}}$ and minimum hamming distance $d_H$ = $2^{m-r}$.
The generator matrix of $\mathcal{R}(r,m)$ is
\[
G_{r,m} = 
\left(
\begin{array}{cc}
    G_{r,m-1} &  G_{r,m-1} \\
    \textbf{0} & G_{r-1,m-1}
\end{array} 
\right), \mbox{ for }1\leq r\leq m-1,
\]
where 
\[
G_{m,m} = 
\left(
\begin{array}{c}
    G_{m-1,m}  \\
    1\ 1\ldots 1\ 0
\end{array}
\right),
\]
$G_{0,m}$ is the of size $1\times 2^m$ with all entries equal to $1$, and $\textbf{0}$ is a zero matrix with $2^{\sum_{i=1}^r\binom{m}{i}}$ - $2^{\sum_{i=1}^{r-1}\binom{m-1}{i}}$ rows and $2^{m-1}$ columns.

In the following theorem, the minimum distance (Definition \ref{kg distance}) is obtained for Reed Muller codes.
In addition, the encoded DNA code, obtained from binary Reed Muller code satisfies various constraints. 

\begin{theorem}
For positive integers $m$, $r$ ($0\leq r\leq m$) and $\ell$, consider $\x,\y\in\Sigma_{DNA}^\ell$ such that 
\begin{itemize}
    \item for $\x^*\in\{\x,\x^c\}$ and $\y^*\in\{\y,\y^c\}$, all $(\x\ \y^*\ \x^*\ \y^*)$ and $(\y\ \x^*\ \y^*\ \x^*)$ DNA strings are $2\ell-1$ conflict free,
    \item each of $(\x\ \y\ \x)$, $(\x\ \y\ \x^c)$, $(\x\ \y^c\ \x)$ and $(\x\ \y^c\ \x^c)$ is free from reverse-complement sub-string(s),
    \item $d_H(\x,\y^{rc})$ = $d_H(\x,\y^r)$ = $\ell$, and
    \item the sum of the $GC$ content of $\x$ and $\y$ will be $\ell$.
\end{itemize} 
For the binary Reed-Muller code $\mathcal{R}(r,m)$, there exists a DNA code $\mathscr{C}_{DNA}(\ell2^m,2^{\sum_{i=1}^r\binom{m}{r}},\ell2^{m-r-1})$ satisfying Hamming, reverse and reverse-complement, where each DNA codeword is $\ell$ conflict free with $GC$ content $(\ell2^{m-1})$ and also free from hairpin like structure.  
\label{reed muller theorem}
\end{theorem}
\begin{proof}
For positive integers $m$ $(>1)$ and $\ell$ $(\leq 2^{m-1})$, consider a binary Reed-Muller code $\mathcal{R}(r,m)$ and the corresponding encoded DNA code $f(\mathcal{R}(r,m))$ for some pair $(\x,\y)\in (\Sigma_{DNA}^\ell)^2$. 
From Theorem \ref{distance preserving theorem}, the codeword length and code size for  $f(\mathcal{R}(r,m))$ will be $\ell2^m$ and $2^{\sum_{i=1}^r\binom{m}{r}}$. 
From Theorem \ref{kg half dis}, the minimum distance $d\geq \ell\lceil d_H/2\rceil$ = $\ell2^{m-r-1}$.
For a positive integer $t$, we denote $\textbf{0}_t$ = $(0\ 0\ldots 0)$ and $\textbf{1}_t$ = $(1\ 1\ldots 1)$, each of length $t$.
Then the binary strings $\textbf{0}_{2^m}$ and $(\textbf{0}_{2^m}\ \textbf{1}_{2^{m-r}})$ will be in $\mathcal{R}(r,m)$.
Therefore $d\leq d(\textbf{0}_{2^m},(\textbf{0}_{2^m}\ \textbf{1}_{2^{m-r}}))$ = $\ell2^{m-r-1}$. 
Hence, $d$ = $\ell2^{m-r-1}$ for the binary $\mathcal{R}(r,m)$.
So, from Theorem \ref{distance preserving theorem}, the minimum Hamming distance for $f(\mathcal{R}(r,m))$ will be $d_H$ = $\ell2^{m-r-1}$.
Now, from the definition of DNA code, the proof follows for Hamming constraint. 
Since the DNA codeword length $\ell2^m$ is even for each pair $m$ and $\ell$, and $d_H(\x,\y^{rc})$ = $d_H(\x,\y^r)$ = $\ell$, therefore, from Theorem \ref{kg reverse theorem} the DNA code $f(\mathcal{R}(r,m))$ satisfies reverse constraint.  
From Theorem \ref{comp dis}, the DNA code $f(\mathcal{R}(r,m))$ follows complement constraint. 
Therefore, from Remark \ref{r rc constraint remark}, the DNA code meets reverse-complement constraint.
From Lemma \ref{kg gen gc cont}, the $GC$ content of each encoded DNA codeword is $(\ell2^{m-1})$.
From Theorem \ref{conflict free theorem}, each encoded DNA codeword is $2\ell-1$ conflict free.
From Theorem \ref{hair pin theorem}, all DNA codewords are free from hairpin like structures.  
\end{proof}
\begin{table}[t]
\caption{List of all encoded DNA strings for $(\x,\y)$ = $(ACT,CTG)$ from $[7,4,3]$ binary Hamming code.}
\begin{center}
\begin{tabular}{|c|c|}
\hline
   $[7,4,3]$ Hamming Code   &  Encoded DNA Code \\  \hline
   $0000000$  & $ATACGCATACGCATACGCATA$      \\  \hline
   $1110000$  & $TATCGCTATGCGTATGCGTAT$  \\  \hline
   $1001100$  & $TATGCGTATCGCTATGCGTAT$  \\  \hline
   $0111100$  & $ATAGCGATAGCGATACGCATA$  \\  \hline
   $0101010$  & $ATAGCGTATCGCATAGCGTAT$  \\  \hline
   $1011010$  & $TATGCGATAGCGTATCGCATA$  \\  \hline
   $1100110$  & $TATCGCATACGCTATCGCATA$  \\  \hline
   $0010110$  & $ATACGCTATGCGATAGCGTAT$  \\  \hline
   $1101001$  & $TATCGCATAGCGTATGCGATA$  \\  \hline
   $0011001$  & $ATACGCTATCGCATACGCTAT$  \\  \hline
   $0100101$  & $ATAGCGTATGCGATACGCTAT$  \\  \hline
   $1010101$  & $TATGCGATACGCTATGCGATA$  \\  \hline
   $1000011$  & $TATGCGTATGCGTATCGCTAT$  \\  \hline
   $0110011$  & $ATAGCGATACGCATAGCGATA$  \\  \hline
   $0001111$  & $ATACGCATAGCGATAGCGATA$  \\  \hline
   $1111111$  & $TATCGCTATCGCTATCGCTAT$  \\  \hline
\end{tabular}
\end{center}
\label{Hanning codewords table}
\end{table}
\section{Conclusion}
\label{sec:conclusion}
In this article, two different approaches (computational and algebraic) have been used to construct DNA codes. Computational approach improves the lower bounds on the size of the DNA codes in many cases from previous study under a new constraint (generalization of homo-polymer constraint) apart from the popular constraint. In algebraic approach, an isometry between binary space and DNA codes is proposed and used to construct many new classes of DNA codes. The new codes are also free from hair-pin like secondary structures.

It would be an interesting future task to find bounds on DNA codes with the new constraint in mind and constructing optimal codes meeting those bounds. Extending the isometry from binary to $q$-ary case will also be an interesting future task.  

\bibliographystyle{IEEEtran}
\bibliography{IEEE_TCOM}
\begin{table*}[t]
\caption{Pairs $(\x,\y)\in\Sigma^\ell$ such that (i) $d_H(\x,\y)$ = $d_H(\x,\y^{rc})$ = $d_H(\x,\y^r)$ = $\ell$, (ii) $GC$-content sum of $\x$ and $\y$ is $\ell$, and (iii) each DNA string in the set $\{(\x\ \y^*\ \x^*\ \y^*), (\y\ \x^*\ \y^*\ \x^*): \x^*\in\{\x,\x^c\} \mbox{ and } \y^*\in\{\y,\y^c\}\}$ is $2\ell-1$ conflict free. 
}
\begin{center}\scriptsize
\begin{tabular}{|c|c|l|}
\hline
$\ell$ & $\#(\x,\y)$ & \hspace{6.5cm}$(\x,\y)$\\ \hline
$3$ & $8$ & $(ATA,CGC)$, $(ATA,GCG)$, $(CGC,ATA)$, $(CGC,TAT)$, $(GCG,ATA)$, $(GCG,TAT)$, $(TAT,CGC)$,  \\ 
    &      & $(TAT,GCG)$ \\ \hline
$4$ & $32$ & $(ATCA,CGAC)$, $(GTCA,CGAT)$, $(ATGA,GCAG)$, $(CTGA,GCAT)$, $(ACTA,CAGC)$, $(GCTA,CAGT)$, \\  
    &      & $(AGTA,GACG)$, $(CGTA,GACT)$,$(CGAC,ATCA)$, $(TGAC,ATCG)$, $(CAGC,ACTA)$, $(TAGC,ACTG)$,  \\  
    &      & $(ATGC,TCAG)$, $(CTGC,TCAT)$, $(AGTC,TACG)$, $(CGTC,TACT)$, $(GCAG,ATGA)$, $(TCAG,ATGC)$, \\  
    &      & $(GACG,AGTA)$, $(TACG,AGTC)$, $(ATCG,TGAC)$, $(GTCG,TGAT)$, $(ACTG,TAGC)$, $(GCTG,TAGT)$, \\ 
    &      & $(GCAT,CTGA)$, $(TCAT,CTGC)$, $(CGAT,GTCA)$, $(TGAT,GTCG)$, $(GACT,CGTA)$, $(TACT,CGTC)$, \\  
    &      & $(CAGT,GCTA)$, $(TAGT,GCTG)$ \\ \hline
 $5$ & $112$ & $(ACGCA,CTATC)$, $(ACGCA,GATAG)$, $(CTGCA,GACAT)$, $(GATCA,CGAGT)$, $(GCTCA,CTGAT)$, \\
    &      & $(GCTCA,CTAGT)$, $(AGCGA,CATAC)$, $(AGCGA,GTATG)$, $(GTCGA,CAGAT)$, $(CATGA,GCACT)$, \\
    &      & $(CGTGA,GTCAT)$, $(CGTGA,GTACT)$,$(ATCTA,GCGAC)$, $(ATCTA,CGAGC)$, $(ATCTA,CGTGC)$, \\  
    &      & $(ATCTA,GCACG)$, $(ATCTA,CAGCG)$, $(ATCTA,GCTCG)$, $(GTCTA,CGACT)$, $(GTCTA,CAGCT)$, \\  
    &      & $(ATGTA,CGAGC)$, $(ATGTA,GACGC)$, $(ATGTA,CGTGC)$, $(ATGTA,CGCAG)$, $(ATGTA,GCACG)$, \\  
    &      & $(ATGTA,GCTCG)$, $(CTGTA,GCAGT)$, $(CTGTA,GACGT)$, $(TAGAC,AGTCG)$, $(TAGAC,AGCTG)$,  \\  
    &      & $(TCGAC,ATCTG)$, $(CATAC,AGCGA)$, $(CATAC,TCGCT)$, $(AGTAC,TCACG)$, $(CGAGC,TATCA)$,  \\  
    &      & $(CGAGC,ATCTA)$, $(CGAGC,ATGTA)$, $(CGAGC,TACAT)$, $(CGAGC,TAGAT)$, $(CGAGC,ACTAT)$, \\  
    &      & $(TGAGC,ATCAG)$, $(TGAGC,ACTAG)$, $(AGTGC,TCATG)$, $(AGTGC,TACTG)$, $(CGTGC,TCATA)$, \\  
    &      & $(CGTGC,ATCTA)$, $(CGTGC,ATGTA)$, $(CGTGC,TACAT)$, $(CGTGC,TAGAT)$, $(CGTGC,ATACT)$, \\  
    &      & $(TGATC,ACTCG)$, $(CTATC,ACGCA)$, $(CTATC,TGCGT)$, $(ACGTC,TACAG)$, $(ATGTC,TGCAG)$, \\  
    &      & $(ATGTC,TGACG)$, $(TACAG,ACTGC)$, $(TACAG,ACGTC)$, $(TGCAG,ATGTC)$, $(GATAG,ACGCA)$, \\  
    &      & $(GATAG,TGCGT)$, $(ACTAG,TGAGC)$, $(GCACG,TATGA)$, $(GCACG,ATCTA)$, $(GCACG,ATGTA)$, \\  
    &      & $(GCACG,TACAT)$, $(GCACG,TAGAT)$, $(GCACG,AGTAT)$, $(TCACG,ATGAC)$, $(TCACG,AGTAC)$, \\  
    &      & $(ACTCG,TGATC)$, $(ACTCG,TAGTC)$, $(GCTCG,TGATA)$, $(GCTCG,ATCTA)$, $(GCTCG,ATGTA)$, \\  
    &      & $(GCTCG,TACAT)$, $(GCTCG,TAGAT)$, $(GCTCG,ATAGT)$, $(TCATG,AGTGC)$, $(GTATG,AGCGA)$, \\  
    &      & $(GTATG,TCGCT)$, $(AGCTG,TAGAC)$, $(ATCTG,TCGAC)$, $(ATCTG,TCAGC)$, $(GACAT,CTGCA)$, \\  
    &      & $(GACAT,CGTCA)$, $(TACAT,CGAGC)$, $(TACAT,CGTGC)$, $(TACAT,GCGTC)$, $(TACAT,GCACG)$, \\  
    &      & $(TACAT,CTGCG)$, $(TACAT,GCTCG)$, $(CAGAT,GTCGA)$, $(CAGAT,GCTGA)$, $(TAGAT,CGAGC)$, \\  
    &      & $(TAGAT,GTCGC)$, $(TAGAT,CGTGC)$, $(TAGAT,GCACG)$, $(TAGAT,GCTCG)$, $(TAGAT,CGCTG)$, \\  
    &      & $(GCACT,CATGA)$, $(GCACT,CAGTA)$, $(GTACT,CGTGA)$, $(CAGCT,GTCTA)$, $(TCGCT,CATAC)$, \\  
    &      & $(TCGCT,GTATG)$, $(CGAGT,GATCA)$, $(CGAGT,GACTA)$, $(CTAGT,GCTCA)$, $(GACGT,CTGTA)$, \\  
    &      & $(TGCGT,CTATC)$, $(TGCGT,GATAG)$ \\ \hline
\end{tabular}
\end{center}
\label{hairpin l table}
\end{table*}

\end{document}